\documentclass[11pt]{article}

\usepackage{amssymb}
\usepackage{amsmath}
\input{epsf.sty}  

\newtheorem{proposition}{\bf Proposition}
\newtheorem{example}{\bf Example}
\newtheorem{remark}{\bf Remark}

\newenvironment{proof}{
\begin{trivlist}
\item[\hspace{\labelsep}{\bf\noindent Proof. }] }{\par\hfill\end{trivlist}
\par}

\title{\LARGE \bf 
Analysis of a growth model inspired by Gompertz and Korf laws, 
and an analogous birth-death process}

\author{
{\sc Antonio Di Crescenzo}\footnote{
Dipartimento di Matematica, Universit\`a degli Studi di Salerno,
Via Giovanni Paolo II, 132, 84084 Fisciano (SA), Italy, email: adicrescenzo@unisa.it}
\and 
{\sc Serena Spina}\footnote{
Dipartimento di Matematica, Universit\`a degli Studi di Salerno,
Via Giovanni Paolo II, 132, 84084 Fisciano (SA), Italy, email: sspina@unisa.it}
}
 
\date{\normalsize 
\bf First published in {\em Mathematical Biosciences}, Vol.\ 282, p.\ 121--134\\
\copyright\ 2016 by Elsevier
}

\begin{document}

\maketitle

\begin{abstract}
We propose a new deterministic growth model which captures certain features of both the 
Gompertz and Korf laws. We investigate its main properties, with special attention to the 
correction factor, the relative growth rate, the inflection point, the maximum specific growth 
rate, the lag time and the threshold crossing problem. Some data analytic examples and their performance are also 
considered. Furthermore, we  study a stochastic counterpart of the proposed model, that is a 
linear time-inhomogeneous birth-death process whose mean behaves as the deterministic one. 
We obtain the transition probabilities, the moments and the population ultimate extinction 
probability for this process. We finally treat the special case of a simple birth process, 
which better mimics the proposed growth model.


\medskip\noindent
{\em Maths Subject Classification:}   92D25; 60J80; 60J85.
\end{abstract}
   
%
\section{Introduction and background}
\label{intro}
%
Constructing mathematical models of real evolutionary phenomena is relevant in many fields. 
It is well known that the exponential curve is a basic model for the description of growth 
without constrains. It is solution of the differential equation
\begin{equation}\label{exp_gr}
\frac{dN(t)}{dt}=r N(t),\qquad t>0,
\end{equation}
where $N(t)$ represents the population size and $r> 0$ is the growth rate. However, limitations of the 
availability of nutrients and habitat, for example, make the exponential curve not appropriate for the 
description of long-term growth. Alternative formulations of growth models, mainly involving regulatory 
effects, have been proposed in the past to take into account that the population growth slows down 
when the resources run out.
\par
One of the first generalizations of (\ref{exp_gr}) is the logistic equation 
\begin{equation}\label{logistic}
\frac{dN(t)}{dt}=r N(t)\left[1- \frac{N(t)}{C}\right],\qquad t>0,
\end{equation}
where $C>0$ represents the so called \lq\lq carrying capacity\rq\rq, which is the  maximum population 
size that the environment can sustain indefinitely. Several investigations contributed to the development 
of growth curves to describe specific behavior of the population dynamics, such as the logistic model 
and its generalizations (see, for example \cite{Bhowmick}, \cite{Talkington}, \cite{Tsoularis}, \cite{Zwietering}).
\par
Specifically, in \cite{Tsoularis} a generalized form of the logistic growth has been introduced through the 
following equation:
\begin{equation}\label{log_gen}
\frac{dN(t)}{dt}=r N(t)^a \left[1-\left(\frac{N(t)}{C}\right)^b\right]^c,\qquad t>0,
\end{equation}
where $a$, $b$, $c$, $r$ are positive constants, $N(t)$ represents the populations size, and $C>0$ is still  
the carrying capacity. Various well-known models arise as special cases of (\ref{log_gen}), such as the  
exponential, monomolecular, power, generalized von Bertalanffy, specialized von Bertalanffy, Richards, 
Smith, Blumberg, hyperbolic, generic, Schnute models (for details, see \cite{Tsoularis}). Furthermore, in 
various settings the constant growth rate $r$ of equation (\ref{exp_gr}) is substituted by a time-dependent 
rate. When $r$ is replaced by a decreasing exponential function we obtain the Gompertz model of 
population growth (see \cite{Gompertz}), governed by equation 
\begin{equation}\label{gompertz}
\frac{dN_G(t)}{dt}=\alpha e^{-\beta t} N_G(t),\qquad t>0.
\end{equation}
Another suitable choice leads to the Korf type model (see \cite{Korf}), described by
\begin{equation}\label{korf}
\frac{dN_K(t)}{dt}=\alpha t^{-(\beta+1)}N_K(t),\qquad t>0.
\end{equation}
Note that both Gompertz and Korf models can be obtained also by means of special choices of 
the parameters in equation (\ref{log_gen}), as shown in \cite{Tsoularis}.
\par
Another general model for biological growth was proposed by Koya and Goshu \cite{Koya} aiming 
to generalize the most common growth models such as Brody, Von Bertalanffy, Richards, Weibull, 
monomolecular, Mitscherlich, Gompertz, logistic models, among others. In this case, the generalized 
growth function is defined as 
\begin{equation}\label{Koya-Goshu}
N(t)=A_L+(A-A_L)\left\{1-Be^{-k[(t-\mu)/\delta]}\right\}^m,\qquad t>0,
\end{equation}
where $k>0$, $A_L$ is the lower asymptote of $N(t)$, $A=\lim_{t\rightarrow \infty} N(t)$, $\mu$ is the time shift 
(a constant), $\delta$ is the time scale (a constant), $m$ is a shape parameter of the growth 
function ($m\neq 0$), and $B=1-\left(\frac{A_\mu-A_L}{A-A_L}\right)^{1/m}$, where 
$A_\mu\equiv N(\mu)$ is the growth rate parameter. 
\par
More recently, in order to model biological dynamics, a new class of growth curves  has been 
proposed in \cite{Bhowmick} as solution of the following differential equation:
\begin{equation}\label{other_gen}
\frac{dN(t)}{dt}=b e^{-a t}t^c N(t), \qquad t>0,
\end{equation}
this extending both the Gompertz and Korf laws. 
\par
In order to introduce a new deterministic model of population growth which is to some 
extent  related with the Gompertz and Korf laws, hereafter we recall some basic issues. 
Due to (\ref{gompertz}), the Gompertz curve is expressed by:
\begin{equation}\label{gompertz_sol}
 N_G(t)=y  \exp{\left\{\frac{\alpha}{\beta}\left(1-e^{-\beta t}\right)\right\}},\quad t>0,
 \qquad N_G(0)=y>0,
\end{equation}
with $\alpha, \beta >0$, and is used to describe population dynamics in a confined habitat. It is also 
employed to describe plant disease progress \cite{Berger}, for mobile phone uptake \cite{Islam}, and 
as software reliability model \cite{Wood}. Above all, the Gompertz model plays a relevant role especially 
in modeling tumor growth (see \cite{Gompertz}, \cite{Laird_64}, \cite{Laird_65}). Until recent time,   
it is a specific reference in this field (see \cite{Milotti}, for instance), and mainly in its stochastic 
form (see \cite{Albano_06}). 
\par
According to (\ref{korf}), the Korf growth curve is given by
\begin{equation}\label{korf_sol}
N_K(t)=y \exp{\left\{\frac{\alpha}{\beta}\left(1-t^{-\beta}\right)\right\}},
\quad t>0,\qquad N_K(0)=0,
\end{equation}
with $\alpha, \beta, y >0$.  
In \cite{Podrazsky} the Korf growth function is used for the assessment of current and mean annual 
increments of Douglas-fir compared to other tree species (see also \cite{Zeide}). Moreover, in \cite{Sedmak} 
it is used for the construction of domestic yield tables, whereas Torres {\em et al.} \cite{Torres_A} used Korf 
and von Bertalanffy models to fit curves as non-linear effects models, showing that Korf model was superior. 
\par
On the ground of the above mentioned investigations, in this paper we aim to propose a new 
growth model which cannot be obtained as a special case of the models (\ref{log_gen}), 
(\ref{Koya-Goshu}) and (\ref{other_gen}). Moreover, even if our model has the same carrying 
capacity of the Gompertz and Korf laws, it is able to capture different evolutionary dynamics. 
Indeed, the new growth model is suitable to describe a phenomenon which behaves 
as the Gompertz curve at the beginning of the evolution (having the same initial value and 
initial slope), but for large times it grows slower than the Gompertz law, as well as the 
Korf curve. Specific details on mathematical properties and their biological meaning for the 
proposed model are pinpointed in the first part of the paper. 
\par
Until  now we considered deterministic growth process only. However, sto\-chastic fluctuations often are 
essential to describe real growth phenomena. Indeed, some of the previous models have been formulated 
in the stochastic environment too. 
We recall that two examples of density-dependent birth-death processes 
whose means satisfy the logistic and Gompertz equations are given in \cite{Parthasarathy_1990} and
\cite{Parthasarathy_1991}, and a general theory for some non-homogeneous density-dependent birth-death 
processes has been developed in \cite{Tan_Piantadosi} with special applications to stochastic logistic growth. 
Logistic and Gompertz type growth are used in \cite{Aagaard} as example to show the effect of the 
variability introduced into the generalized Poulsen population model by assuming that the random 
number of offspring depends on the population size. 
We recall that birth-death processes are widely considered within applications in ecology, genetics 
and evolution (see, e.g., Crawford and Suchard  \cite{CrawfordSuchard} and references therein). 
Other families of suitable point processes to describe population growth have been 
studied recently in Di Crescenzo  {\em et al.} \cite{DGN}, and in Cha and Finkelstein \cite{ChaF}. 
Moreover, diffusion stochastic processes have been proposed to take into account environmental fluctuations. 
For instance, the random effects of demographic and environmental stochasticity on a tumor immunology 
deterministic model is studied in \cite{Rosenkranz} throughout a stochastic diffe\-rential equation, 
whose solution is a limiting diffusion process.
Furthermore, a stochastic diffusion model related to a reformulation of the Richards growth curve is proposed 
in \cite{Roman_2015}, and a Gompertz-type diffusion process is introduced in \cite{Gutierrez}, by means of 
which bounded sigmoidal growth patterns can be modeled by time-continuous variables. 
\par
In agreement with some of the previously mentioned investigations, 
aiming to extend the new growth law to a more realistic setting characterized by the 
presence of randomness, we construct and study two suitable 
stochastic processes whose means correspond to the proposed growth model: 
(i) a non-homogeneous birth-death process, and (ii) a simple birth process. 
\par
The paper is organized as follows. In Section \ref{sect:model} the new model is introduced and its main 
properties are investigated, with special attention to the correction factor, the relative growth rate, the 
inflection point, the maximum specific growth rate, the lag time and the threshold crossing problem. 
A comparison to the Gompertz and Korf models is also performed  
with purpose of showing novelties of the proposed models, and analogies and differences 
among the three considered laws.  In Section \ref{sect:dataex}, some 
data analytic examples are considered in order to show that the new model can better describe some 
evolutionary phenomena with respect to the other ones. Section \ref{sect:linearBD} is devoted to a  linear 
time-inhomogeneous birth-death process whose mean behaves as the proposed deterministic model. 
In particular, we analyze the transition probabilities, the mean, the variance and the population extinction 
probability of this model. We also point out that the transition probabilities obtained in Proposition 
\ref{prop:1} correct an expression given in a previous paper. 
Furthermore, in Section \ref{sect:simpleB} we investigate a simple birth process 
mimicking the new growth model, giving attention to its transition probabilities, the mean, the variance, 
the index of dispersion, the coefficient of variation and the first-passage-time problem.
\par
Throughout the paper primes denote derivatives; moreover we set 
$\mathbb{N}=\{1,2,\ldots\}$ and $\mathbb{N}_0=\{0,1,2,\ldots\}$. 
\section{The growth model} \label{sect:model}
A general model for population growth can  often be described by an ordinary differential equation 
of the form 
\begin{equation}\label{eq:nt}
\frac{dN(t)}{dt}=\xi(t) N(t),\qquad t>0,
\end{equation}
where $\xi(t)>0$ is a time-dependent growth rate. Clearly, if 
$$
 \xi(t)= \xi_G(t)=\alpha e^{-\beta t}\quad \hbox{or} \quad \xi(t)=\xi_K(t)= \alpha t^{-(\beta +1)}, 
$$ 
then Eq.\ (\ref{eq:nt}) yields the differential equation (\ref{gompertz}) of the Gompertz growth and that of the Korf growth (\ref{korf}). 
Hereafter we consider an alternative choice of the growth rate, given by 
\begin{equation}\label{eq:xidit}
\xi(t)= \alpha(1+t)^{-(\beta+1)},\qquad t>0,
\end{equation}
so that, from (\ref{eq:nt}) we obtain the ordinary differential equation 
\begin{equation}\label{new}
\frac{dN(t)}{dt}=\alpha(1+t)^{-(\beta+1)}N(t),\qquad t>0.
\end{equation}
Accordingly, we propose the following curve for growth model, 
with $\alpha, \beta>0$, 
\begin{equation}\label{new_sol}
N(t)=y \exp{\left\{\frac{\alpha}{\beta}\left[1-(1+t)^{-\beta}\right]\right\}},\quad t>0,
\qquad N(0)=y>0,
\end{equation}
which is solution of (\ref{new}). 
In brief, such a model is motivated by the need of describing evolutionary dynamics 
characterized by non-zero initial values and approximately linear initial slope, and 
tending to a carrying capacity from below through a long-term power-law growth. 
We recall that various examples of population growth showing power-law behavior have 
been considered in the recent literature (see, for instance, Karev \cite{Karev}). 
\par
It is worth pointing out that the parameters $\alpha$ and $\beta$ are the growth and the decay rates, 
respectively, for the three models given in (\ref{gompertz_sol}), (\ref{korf_sol}) and (\ref{new_sol}). 
\par
The  model proposed in (\ref{new_sol}), as well as the Gompertz and Korf curves, 
stems from a linear equation. Indeed, we purpose to deal with a rather simple 
model that allows a feasible treatment from mathematical and statistical points 
of view, differently from more general equations similar to (\ref{log_gen}). 
\par
Note that equations (\ref{gompertz}), (\ref{korf}) and (\ref{new}) are all linear, characterized by 
growth rates $\xi(t)$ which are decreasing convex and asymptotically vanishing functions. 
Consequently, such common feature yields that the three curves $N_G(t)$, $N_K(t)$ 
and $N(t)$ share some characteristics. Indeed, due to (\ref{gompertz_sol}), 
(\ref{korf_sol}) and (\ref{new_sol}), such three curves are increasing in $t$, and are 
bounded by the same carrying capacity, i.e. 
\begin{equation}\label{carrying_capacity} 
 \lim_{t\rightarrow +\infty} N_G (t)=\lim_{t\rightarrow +\infty} N_K (t)=\lim_{t\rightarrow +\infty} N (t)
 =C\equiv y e^{\alpha/\beta}.
\end{equation}
Moreover, it is not hard to show that 
\begin{equation}
N_K (t)< N (t)< N_G (t) \quad\hbox{ for all $t>0$.} 
\label{eq:ordinam}
\end{equation}
In other terms, for fixed choices of the parameters $y, \alpha, \beta$, 
the new proposed model describes an intermediate growth between the (lower) Korf 
and (upper) Gompertz curves. Furthermore, 
in Remark \ref{rem:1} below we see that for small values of $t$ the curve $N(t)$ 
behaves similarly as $N_G (t)$, whereas for large times it grows similarly as $N_K (t)$, 
since it tends to the carrying capacity polynomially fast. 
\begin{remark}\label{rem:1}
The growth model (\ref{new_sol}) captures certain features of both the Gompertz and Korf 
laws. Indeed, the curve (\ref{new_sol}) behaves as the Gompertz curve (\ref{gompertz_sol}) 
for $t$ close to $0$, being 
$$
 N(0)=N_G(0)=y \quad  \hbox{and} \quad N'(0)=N_G'(0)=\alpha y.
$$
The Korf law (\ref{korf_sol}) has a rather different initial behavior, since $N_K(0)=0$ and 
$N_K'(0)= +\infty$. On the contrary, for large values of $t$, the proposed model behaves 
more similarly to the Korf law. Specifically, even though the three models have the same 
carrying capacity (\ref{carrying_capacity}), the three curves tend to $C$ according 
to different rules. Indeed, from Eqs.\ (\ref{gompertz_sol}), (\ref{korf_sol}) and (\ref{new_sol}) 
we have 
$$
 \frac{\beta}{\alpha}\Big|\log \frac{N_G(t)}{C}\Big|= e^{-\beta t}, 
 \quad  
 \frac{\beta}{\alpha}\Big|\log \frac{N_K(t)}{C}\Big|= t^{-\beta}, 
 \quad  
 \frac{\beta}{\alpha}\Big|\log \frac{N(t)}{C}\Big|= (1+ t)^{-\beta}.
$$  
Hence, for any curve $N_{\star}(t)$ the term $\big|\log \frac{N_{\star}(t)}{C}\big|$ 
tends to 0 when $t\to +\infty$. However, for the Gompertz law this limit is attained exponentially fast, 
whereas for the other two laws it is reached according to power law with exponent $\beta$.  
\end{remark}
\par
In Fig.\ \ref{fig:1}, the three growth curves are shown for some choices of the parameters. 
Note that when $\beta$ increases, the curves (\ref{korf_sol})  and (\ref{new_sol}) tend to the 
Gompertz one. 
\begin{figure}[t]
\begin{center}
\epsfxsize=5.1cm
\epsfbox{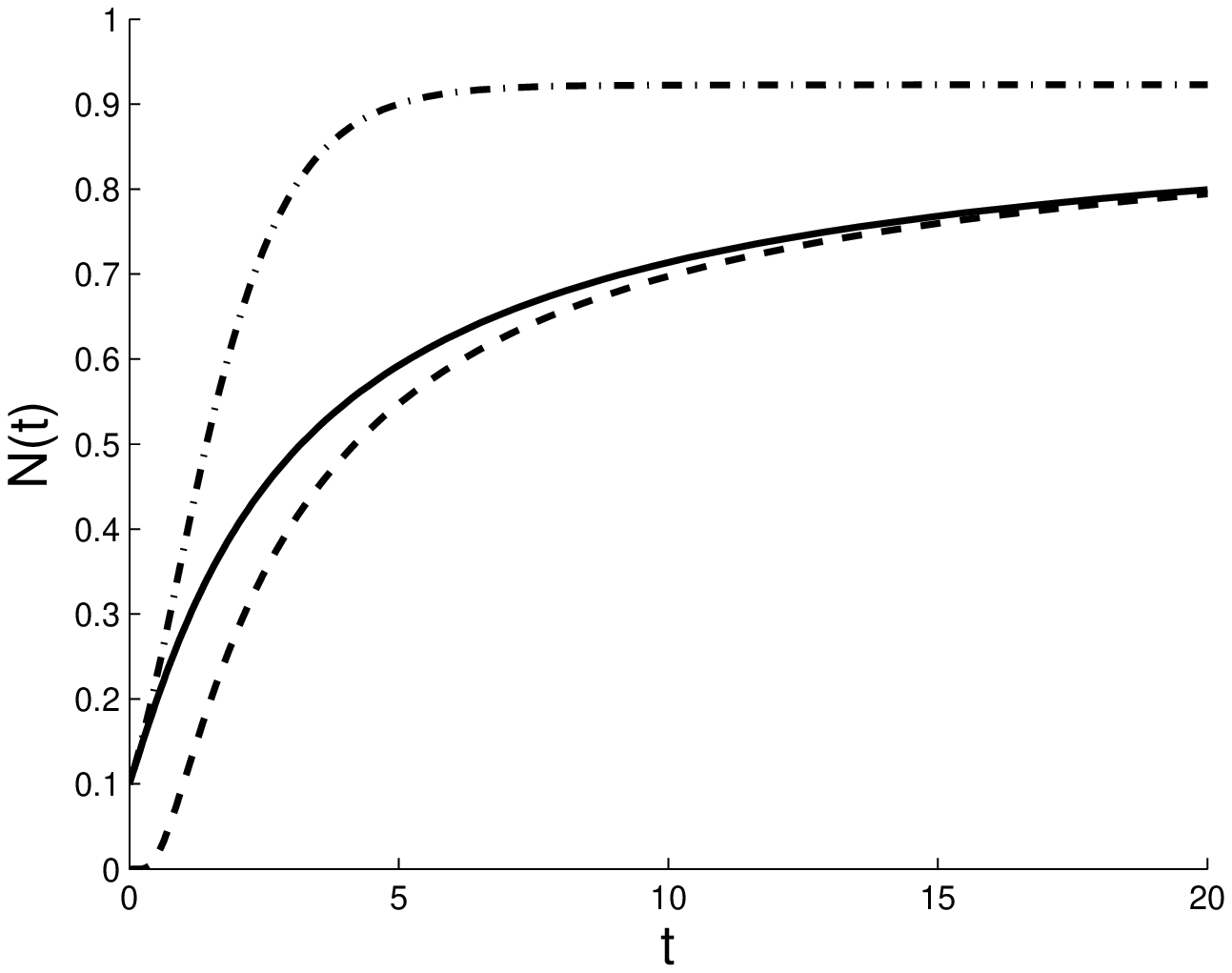}
\epsfxsize=5.1cm
\epsfbox{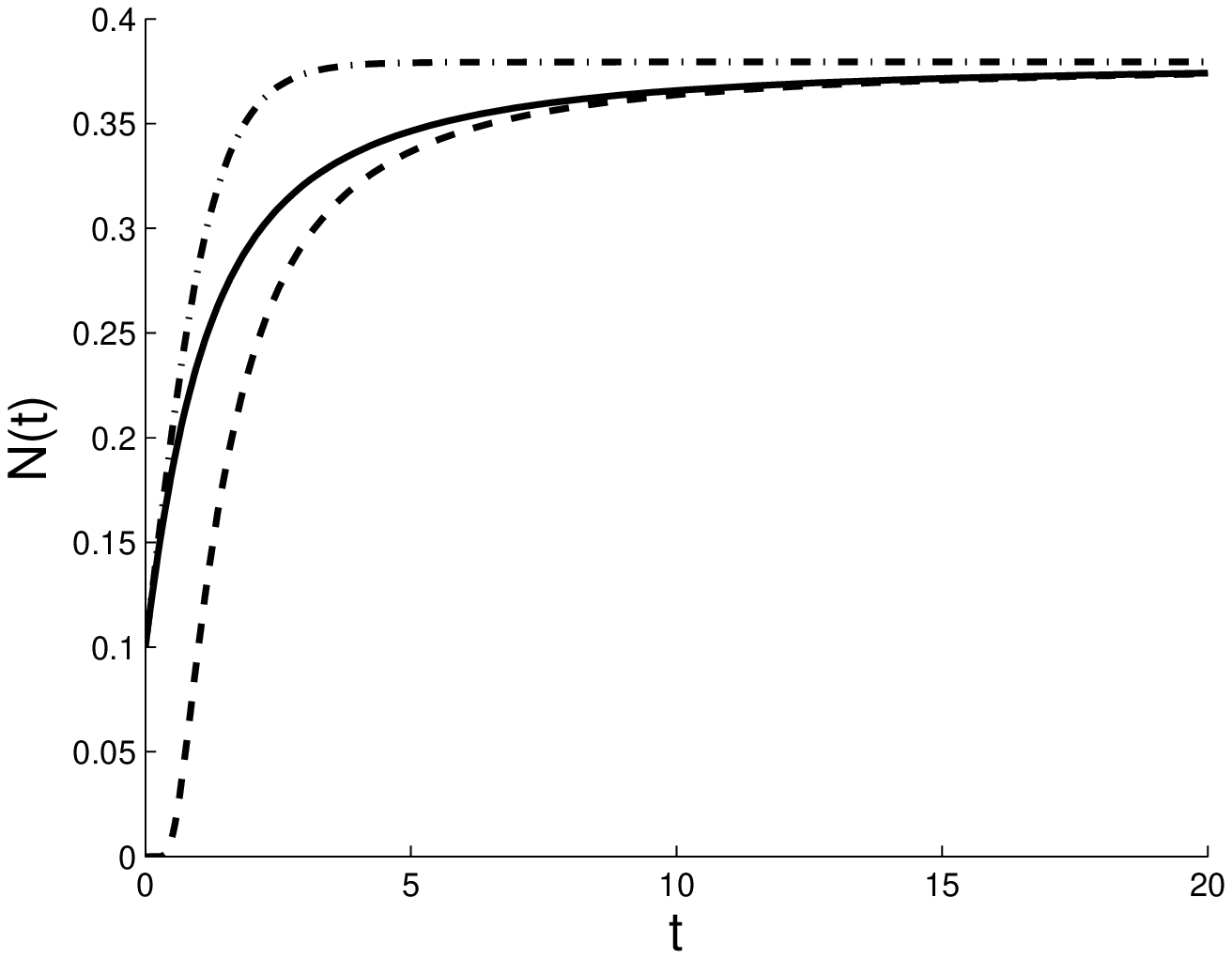}
\epsfxsize=5.1cm
\epsfbox{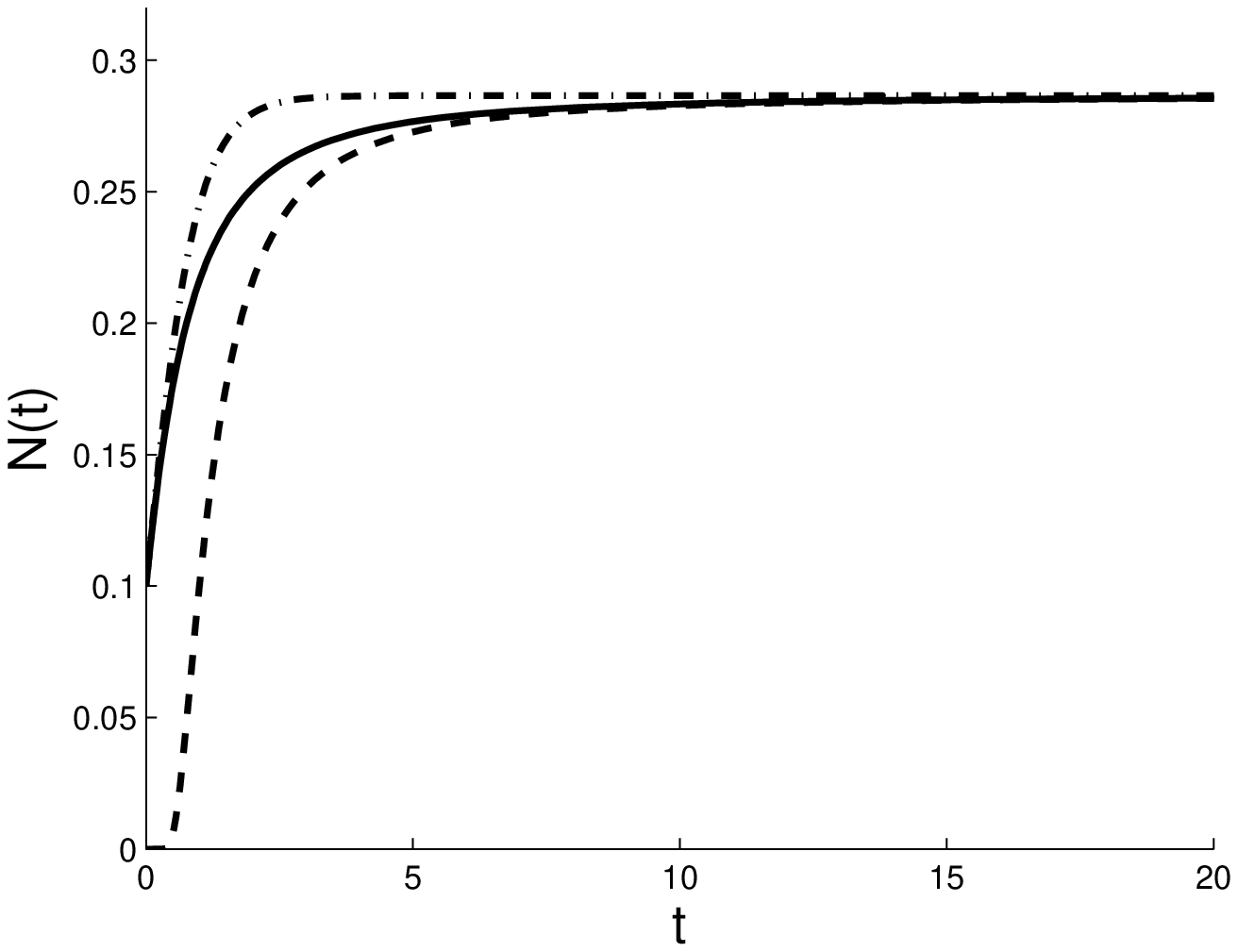}
\caption{The proposed curve (full), Korf curve (dashed), and Gompertz curve (dot-dashed),  
for $y = 0.1$, $\alpha = 2$, $\beta=0.9,\,1.5,\,1.9$, respectively, with $C=0.93,\,0.38,\,0.29$.}
\label{fig:1}
\end{center}
\end{figure}
\par
Some interesting limit behaviors of the curves (\ref{gompertz_sol}), (\ref{korf_sol}) and (\ref{new_sol}) 
are shown in Table 1, for $\alpha \rightarrow 0$, $\alpha \rightarrow +\infty$, $\beta \rightarrow 0$ and 
$\beta \rightarrow +\infty$, with $t>0$. We observe that when the growth rate $\alpha$ goes to zero and 
when $\alpha$ goes to infinity, the three curves have all the same behavior, i.e.\ they tend to $y$ 
in the first case, and to $+\infty$ in the second case. 
Indeed, the population size exhibits a very pronounced growth  when $\alpha$ is very large, since 
the carrying capacity tends to $+\infty$ when $\alpha$ diverges, due to (\ref{carrying_capacity}). 
Moreover, for $\beta \rightarrow \infty$ the curves  
admit the same limit $y$, except for the Korf case when $0<t<1$. 
\par
In the following we analyze some features of the proposed model, and we perform some comparisons 
with the Gompertz and the Korf models. 
%
%
%
\begin{table}[t]    
\begin{center} 
{\footnotesize
\begin{tabular}{cccccc}
   \hline
	  Model $N_{\star}(t)$  &   $N_{\star}'(0)$     &  $\displaystyle\lim_{\alpha \rightarrow 0}N_{\star}(t)$  
	  &   $\displaystyle\lim_{\alpha \rightarrow \infty}N_{\star}(t)$ 
	  &   $\displaystyle\lim_{\beta \rightarrow 0}N_{\star}(t)$
	  &   $\displaystyle\lim_{\beta \rightarrow \infty}N_{\star}(t)$\\
						&             &&&&\\
	\hline
					Eq.\ (\ref{new_sol})       &   $\alpha y$ &  $y$   & $+\infty$ & $y (1+t)^{\alpha}$ & $y$ \\
\footnotesize{$y \exp{\left\{\frac{\alpha}{\beta}\left[1-(1+t)^{-\beta}\right]\right\}}$}	&&&&&\\
	\hline
Eq.\	(\ref{korf_sol})      &   $ +\infty$    &  $y$     & $+\infty$ & $y t^{\alpha}$ & $\left\{
\begin{array}{ll}
0,& 0<t<1\\
y,  & t\geq 1
\end{array}\right.$  \\
	\footnotesize{$y \exp{\left\{\frac{\alpha}{\beta}\left(1-t^{-\beta}\right)\right\}}$}&&&&&\\
	\hline
Eq.\	(\ref{gompertz_sol})  &   $\alpha y$ &  $y$      &  $+\infty$   & $y e^{\alpha t}$  &  $y$  \\
	\footnotesize{$y  \exp{\left\{\frac{\alpha}{\beta}\left(1-e^{-\beta t}\right)\right\}}$}&&&&&\\
	\hline
\end{tabular}
}
\end{center}
\caption{Some characteristics of models (\ref{gompertz_sol}), (\ref{korf_sol}) and (\ref{new_sol}).}
\end{table}
\subsection{The correction factor and the relative growth rate}
Several growth models described by a function $N_\star (t)$ can be expressed in the 
form (see \cite{Talkington})
\begin{equation}\label{generic_form}
\frac{dN_\star(t)}{dt}=\beta N_\star (t) f\left[N_\star(t)\right].
\end{equation}
 The function $f$ is called a size covariate model since it is function of $t$ only through $N$ and it is widely used to represent the density dependent growth; indeed the {\em relative growth rate} $g$ (cf.\ 
\cite{Tsoularis}) is strictly related to $f$ via 
\begin{equation}\label{rgr}
g[N_\star(t)]:=\frac{1}{N_\star(t)}\frac{dN_\star(t)}{dt}\equiv \beta f[N_\star(t)].
\end{equation}
Note that the function $f$ is also called \textit{correction factor} because it expresses the deviation from the classical 
exponential growth (\ref{exp_gr}), for which $f(z)=1$ for all $z\geq 0$.
\par
 A suitable choice of $f(z)$ is
\begin{equation}\label{corr_fun_korf_new}
f(z)=  f_N(z):= \frac{\alpha}{\beta}\left(1-\frac{\beta}{\alpha}\log\frac{z}{y}\right)^{1+1/\beta},
\qquad z>0,
\end{equation}
which corresponds to the family of models given by
$$
N_{\star}(t)=y e^{{\alpha}/{\beta}}\left[1-(t+D)^{-\beta}\right],\qquad t>0,
$$
where $D$ is a constant. Note that if $D=1$ one obtains the model (\ref{new_sol}), whereas 
$D=0$ yields the Korf model (\ref{korf_sol}). The Gompertz model is obtained by choosing  
\begin{equation}\label{corr_fun_gomp}
f(z)= f_G(z):=\frac{\alpha}{\beta}- \log \frac{z}{y}, \qquad z>0,
\end{equation}
this leading to the following family  of population laws: 
$$
N(t)=y e^{ {\alpha}/{\beta}}\exp{\left\{-e^{-\beta t}-D\right\}},\qquad t>0,
$$
for $D=\alpha/\beta-1$. In Fig.\ 2 we compare the correction factor $f(z)$ for the models 
(\ref{gompertz_sol}), (\ref{korf_sol}) and (\ref{new_sol}), plotted against $\log z$. 
We note that the correction factor of the considered models are decreasing in $\beta$. 
Furthermore, from  (\ref{corr_fun_korf_new}) and (\ref{corr_fun_gomp}), for 
$C\equiv y e^{\alpha/\beta}$ we have  
$$
 f_N(z)= 1 \quad \hbox{for  }z= C \, e^{-\left(\frac{\alpha}{\beta}\right)^{1/(\beta+1)}}<C;
 \qquad f_G(z)\to 1 \quad \hbox{for  }z\to C.
$$ 
Hence, for both the new model and the Korf model the correction factor is equal to that of the 
exponential growth for a given value of the population size $z$ (for $z$ smaller than the carrying capacity $C$). On the contrary, for the  Gompertz model the correction factor tends to that of the exponential growth when  $z$ tends to $C$. 

%
\begin{figure}[t] 
\begin{center}
\epsfxsize=5.8cm
\epsfbox{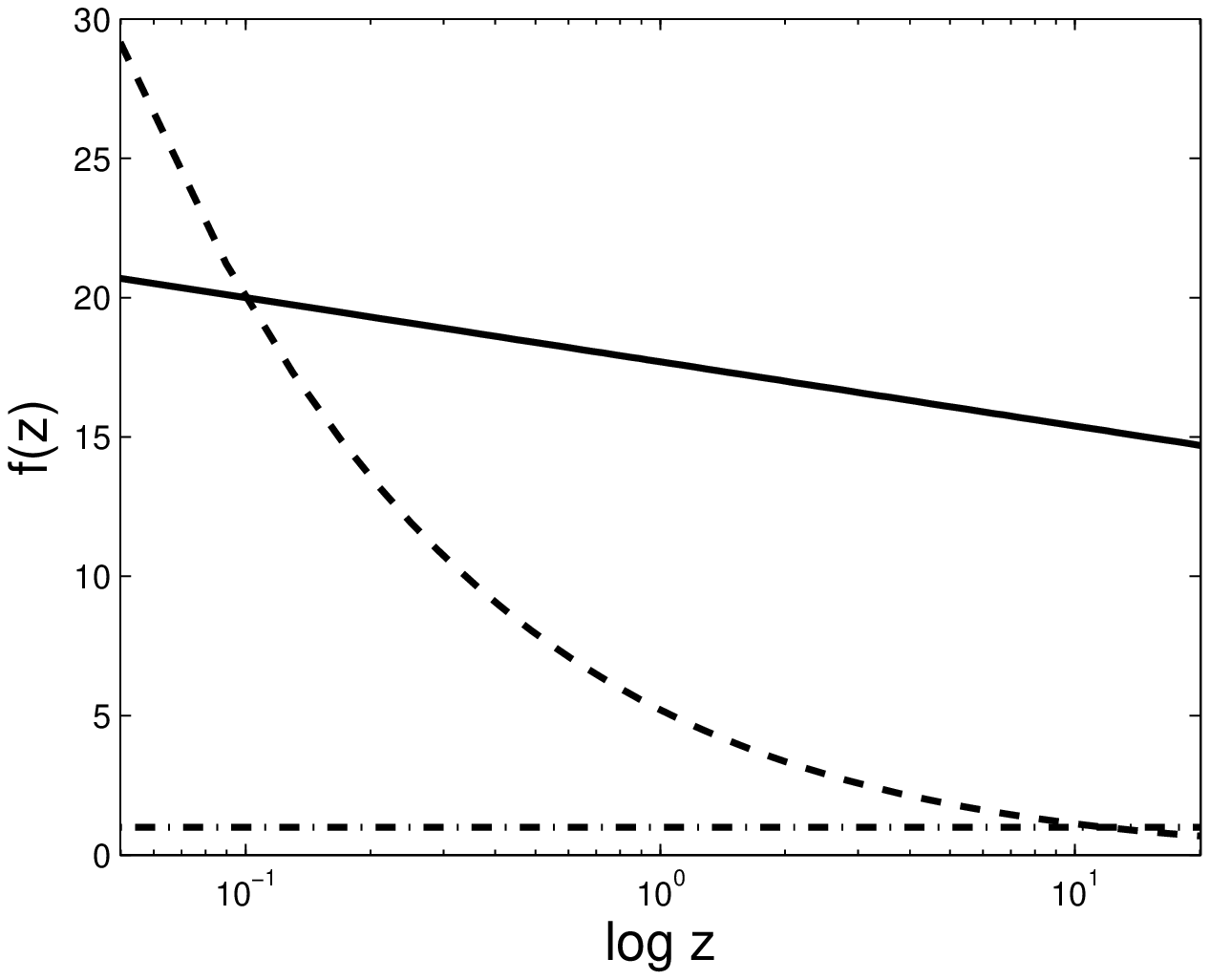} 
\epsfxsize=5.8cm
\epsfbox{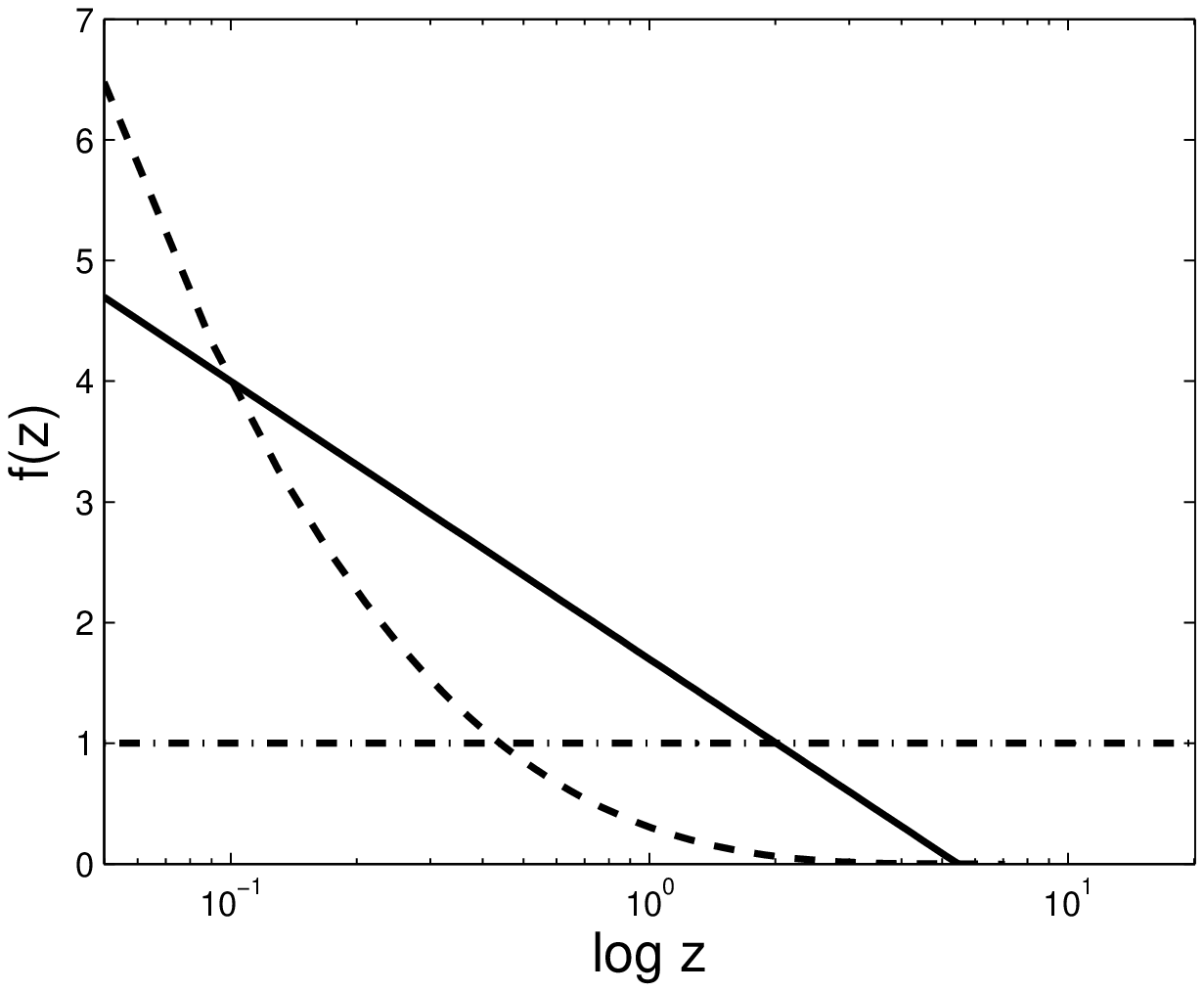}	
\caption{
The correction factor for models (\ref{gompertz_sol}) (solid line), 
(\ref{korf_sol}) and (\ref{new_sol}) (dashed line), for  $y = 0.1$, $\alpha = 2$, $\beta=0.1$ (left), 
$\beta=0.5$ (right), compared with $f(z)=1$ (dot-dashed line).}
\end{center}
\end{figure}
%
%
%
\par
Regarding the relative growth rate $g$ given in (\ref{rgr}),  
for models (\ref{korf_sol}) and (\ref{new_sol}) we have
$g'(z)=-\frac{1+\beta}{z}\left(1-\frac{\beta}{\alpha}\log\frac{z}{y}\right)^{1/\beta}$, 
whereas, due to (\ref{corr_fun_gomp}), for the 
Gompertz model (\ref{gompertz_sol}) we have $g'(z)=-\beta/z$. Since 
$0< z<y e^{\alpha/\beta}=C$, the relative growth rate $g(z)$ is decreasing for all the 
considered models, and reaches the minimum at the carrying capacity $C$. 
%
\subsection{The inflection point}
Let us now focus on the inflection point of the growth model (\ref{new_sol}). 
Clearly, this is of high interest in population growth since for sigmoidal curves 
such point expresses the instant when the growth rate is maximum. 
From (\ref{new}) we have
$$
 \frac{d^2N(t)}{dt^2}=\alpha (1+t)^{-(\beta+2)} N(t)\left[\alpha(1+t)^{-\beta}-\beta-1\right], 
 \qquad t>0.
$$
Hence, if $\alpha\leq \beta+1$ the curve (\ref{new_sol}) has downward concavity for all $t$, 
whereas if $\alpha>\beta +1$ then $N(t)$ is sigmoidal with inflection point
\begin{equation}\label{infl_point_new}
t_N=\left(\frac{\alpha}{\beta+1}\right)^{\frac{1}{\beta}}-1.
\end{equation}
The population at the inflection point  is given by
%
%
$$N(t_N)=ye^{-1-\frac{1}{\beta}} e^{\frac{\alpha}{\beta}}= C e^{-1-\frac{1}{\beta}},$$
%
where $C$ is the carrying capacity, given in (\ref{carrying_capacity}).
\par
In Table 2, the inflection points ($t_N,\,t_K,\,t_G$) and the population at the inflection points 
($N(t_N),\,N_K(t_K),\,N_G(t_G)$) are shown for the three models (\ref{new_sol}), (\ref{korf_sol}) 
and (\ref{gompertz_sol}), respectively. Note that
$$
N_K(t_K)=N(t_N)<N_G(t_G).
$$
Hence, whereas for a fixed time $t$ the considered growth curves are ordered 
according to (\ref{eq:ordinam}), the latter equation shows that the new model and the Korf 
model evaluated at the inflection points have identical population size, which is smaller 
than that of the  Gompertz model. 
\par
Let us now analyze some interesting limits of the population at the inflection point. 
From Table 2 we have, with $N_K(t_K)=N(t_N)$,  
$$
\lim_{\alpha \rightarrow 0} N(t_N)=y e^{-1} e^{-\frac{1}{\beta}},
\qquad \lim_{\beta \rightarrow +\infty} N(t_N)=y e^{-1},
$$
$$
\lim_{\beta \rightarrow 0} N(t_N) =\left\{
\begin{array}{ll}
+\infty,  &\qquad \alpha>1\\
y e^{-1}, &\qquad  \alpha=1  \\
0,&\qquad \alpha<1,
\end{array}\right.
\qquad
\lim_{\alpha \rightarrow +\infty} N(t_N)=+\infty,
$$
whereas for the Gompertz model it results
$$
\lim_{\alpha \rightarrow 0} N_G(t_G) =\lim_{\beta \rightarrow +\infty} N_G(t_G) =y e^{-1},
$$
$$
\lim_{\beta \rightarrow 0} N_G(t_G) =\lim_{\alpha \rightarrow +\infty} N_G(t_G) =+\infty.
$$
%
%
\begin{table}[t]    
\begin{center}
\small{
%
	%
\begin{tabular}{cccc}
   \hline
  Model    &     Inflection   &  Population at the     \\
	         &       point      &  inflection point       \\
	\hline
						Eq. (\ref{new_sol})     &     $t_N=\left(\frac{\alpha}{\beta+1}\right)^{1/\beta}-1$, \footnotesize{$\alpha>\beta+1$} & $N(t_N)=C e^{-1-1/\beta}$ \\
	& \footnotesize{$N$ is concave for $\alpha\leq\beta+1$} &           \\
	\hline
	Eq. (\ref{korf_sol})      &     $t_K=\left(\frac{\alpha}{\beta+1}\right)^{1/\beta}$     &  $N_K(t_K)=C e^{-1-1/\beta}$  \\	
	\hline
	Eq. (\ref{gompertz_sol})  &     $t_G=\frac{1}{\beta}\log\left(\frac{\alpha}{\beta}\right)$, \footnotesize{$\alpha>\beta$}                                                       &    $N_G(t_G)=C e^{-1}$         \\
	        &    \footnotesize{$N_G$ is concave for $\alpha\leq\beta$} &     \\
\hline
\end{tabular}
}
\caption{The inflection points and the population sizes at the inflection points are shown 
for models (\ref{new_sol}), (\ref{korf_sol}) and (\ref{gompertz_sol}).}
\end{center}
\end{table}
%
\subsection{The maximum specific growth rate and the lag time}
In several fields it is interesting to investigate a growth curve in proximity of the inflection point, 
by approximating linearly the curve in that point. 
This is typically of interest in phenomena that exhibit lag, growth, and asymptotic phases, 
such as the growing process of length or mass of some organisms or populations 
(see Zwietering {\em et al.}\ \cite{Zwietering} for details). For a generic growth curve 
$N_\star(t)$ we introduce the {\em maximum specific growth rate} $\mu_{\star}$, defined as 
the coefficient of the tangent to the curve in the inflection point $t_\star$, i.e.
\begin{equation}\label{max_sp_gr_rate}
\mu_\star=\frac{dN_\star(t)}{dt}\Big|_{t=t_\star}.
\end{equation}
Moreover, $\lambda_\star$ denotes the lag time  defined as the $x$-axis intercept of this tangent. 
Specifically, sigmoidal functions describing evolutionary phenomena show a phase 
in which the specific growth rate starts at a zero value and then accelerates to a maximal 
value $\mu$ in a certain period of time, resulting in a lag time $\lambda$. The value of 
$\mu$ is given by the slope of the line when the grow is exponential.
\par
For the  model (\ref{new_sol}), for $\alpha>\beta+1$, recalling (\ref{new}) and (\ref{infl_point_new}), 
due to (\ref{max_sp_gr_rate}) the maximum specific growth rate is  
\begin{equation}\label{mu_N}
 \mu_N=y\frac{(\beta+1)^{1/\beta+1}}{\alpha^{1/\beta}}\exp{\left\{\frac{\alpha-1}{\beta}-1\right\}}.
\end{equation}
Moreover, recalling the expressions of $t_N$ and $N(t_N)$ given in Table 2, the tangent curve 
in $\left(t_N,\,N(t_N)\right)$  is
\begin{equation}\label{tang_new}
 n=\mu_N t+y\exp{\left\{\frac{\alpha-1}{\beta}-1\right\}}
 \left[\frac{(\beta+1)^{1/\beta+1}}{\alpha^{1/\beta}}-\beta\right],
\end{equation}
with $\mu_N$ expressed in (\ref{mu_N}). (For notation clarity, we denote by $n$ the $y$-axis.)  
The lag time $\lambda_N$ for the model (\ref{new_sol}) is the $t$-axis intercept of 
(\ref{tang_new}), that is 
$$
\lambda_N=\beta\frac{\alpha^{1/\beta}}{(\beta+1)^{1/\beta+1}}-1;
$$
note that $\lambda_N$ is positive if, and only if, $\alpha >(\beta+1)^{\beta+1}/\beta^{\beta}$. 
\begin{table}[t]    
\begin{center}
%
{\footnotesize
\begin{tabular}{cccc}
   \hline
	&&&\\
$N_{\star}(t)$   &   $\mu_\star$                    &  Intercept of  &  $\lambda_\star$ \\
						&&                                      tangent curve &\\
	\hline
	         &  &&                                                   \\
 $N(t)$   &  $y\frac{(\beta+1)^{\frac{1}{\beta}+1}}{\alpha^{\frac{1}{\beta}}}e^{\frac{\alpha-1}{\beta}-1}$   & $ye^{\frac{\alpha-1}{\beta}-1}\left[\frac{(\beta+1)^{1/\beta+1}}{\alpha^{1/\beta}}-\beta\right]$  & $\beta\frac{\alpha^{1/\beta}}{(\beta+1)^{1/\beta+1}}-1$  \\
	{}  &  {\tiny ($\alpha>\beta+1$)} &  {\tiny ($\alpha>\beta+1$)} & {\tiny $\left(\alpha >(\beta+1)^{\beta+1}/\beta^{\beta}\right)$}\\   
	\hline
	&&&\\
	$N_K(t)$     &  $y\frac{(\beta+1)^{\frac{1}{\beta}+1}}{\alpha^{\frac{1}{\beta}}}e^{\frac{\alpha-1}{\beta}-1}$     &  $-y \beta e^{\frac{\alpha-1}{\beta}-1}$   &    $\lambda_N+1$        \\
	&&&\\
	\hline
	&&&\\
	$N_G(t)$  &  $y \beta e^{\alpha/\beta-1}$ &  $y e^{\alpha/\beta-1}\left(1-\log \frac{\alpha}{\beta}\right)$  & $\frac{1}{\beta}\log \frac{\alpha}{\beta}-\frac{1}{\beta}$\\	
	{\tiny($\alpha>\beta$)}&&&   \\
	\hline
	\end{tabular}
	}
\end{center}
\caption{The maximum specific growth rate, the intercept of tangent line  and the lag 
time for the models (\ref{new_sol}), (\ref{korf_sol}) and (\ref{gompertz_sol}).}
\end{table}
\par
For comparison purposes, Table 3 shows the maximum specific growth rate, the 
intercept of the tangent curve and the lag time for the three models (\ref{gompertz_sol}), 
(\ref{korf_sol}) and (\ref{new_sol}). The tangent lines (dotted curves) and the lag times 
(asterisk points) for the three models are plotted in Fig.\ 3. 
\begin{figure}[t]
\begin{center}
 model (\ref{new_sol}) \hspace{5cm}  model (\ref{korf_sol}) \\
 \epsfxsize=5.8cm
\epsfbox{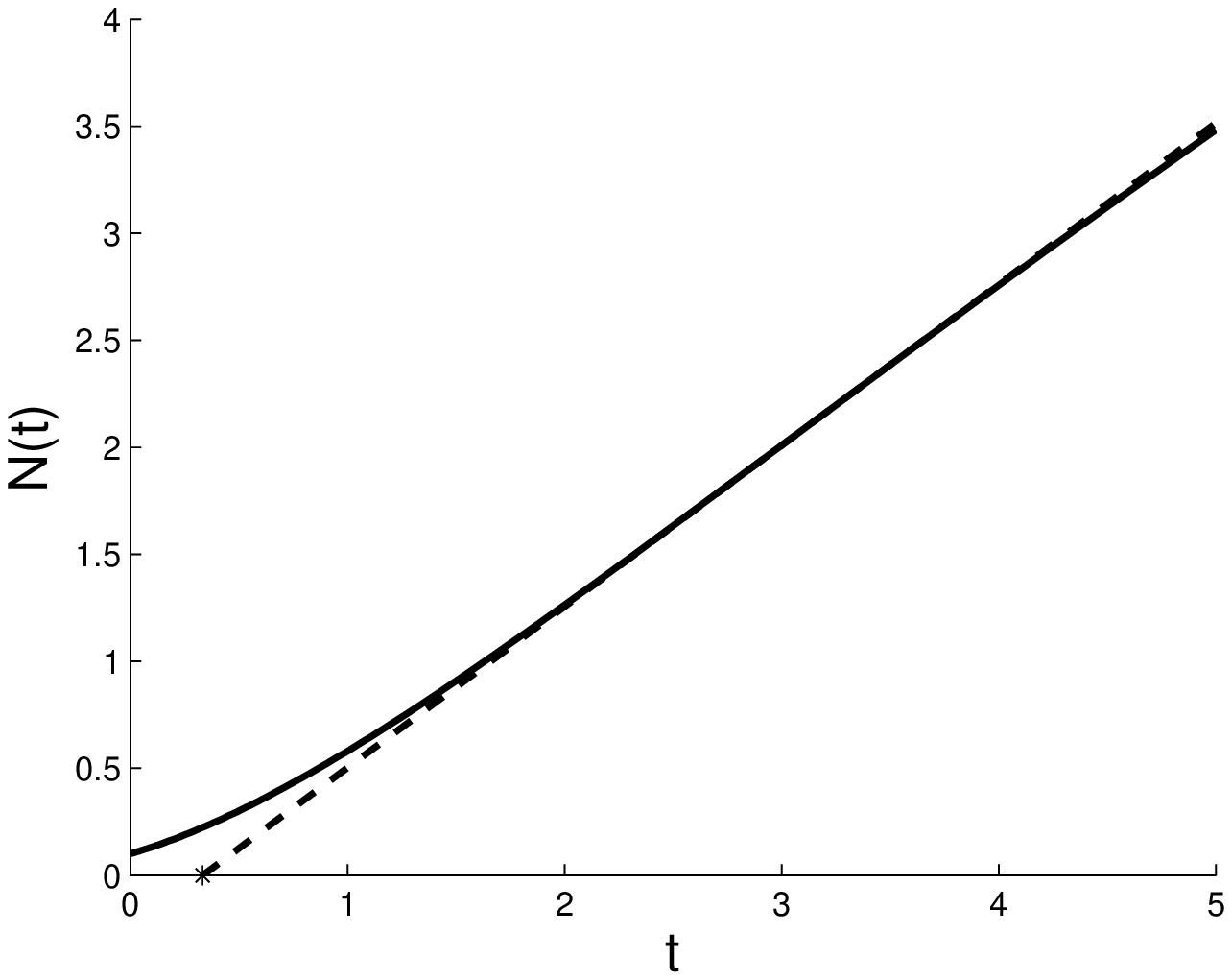}
\epsfxsize=5.8cm
\epsfbox{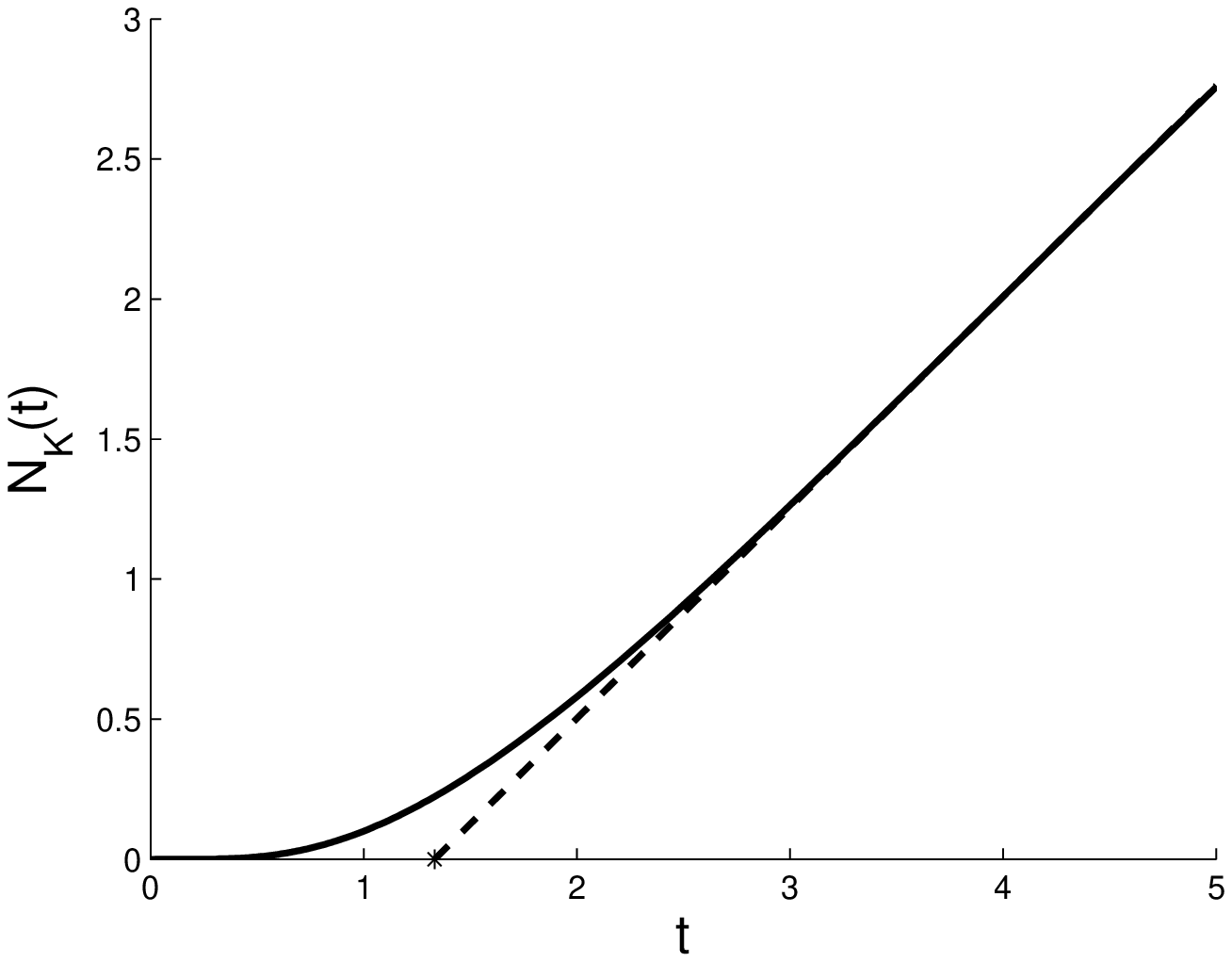}	\\
 model (\ref{gompertz_sol}) \\
\epsfxsize=5.8cm
\epsfbox{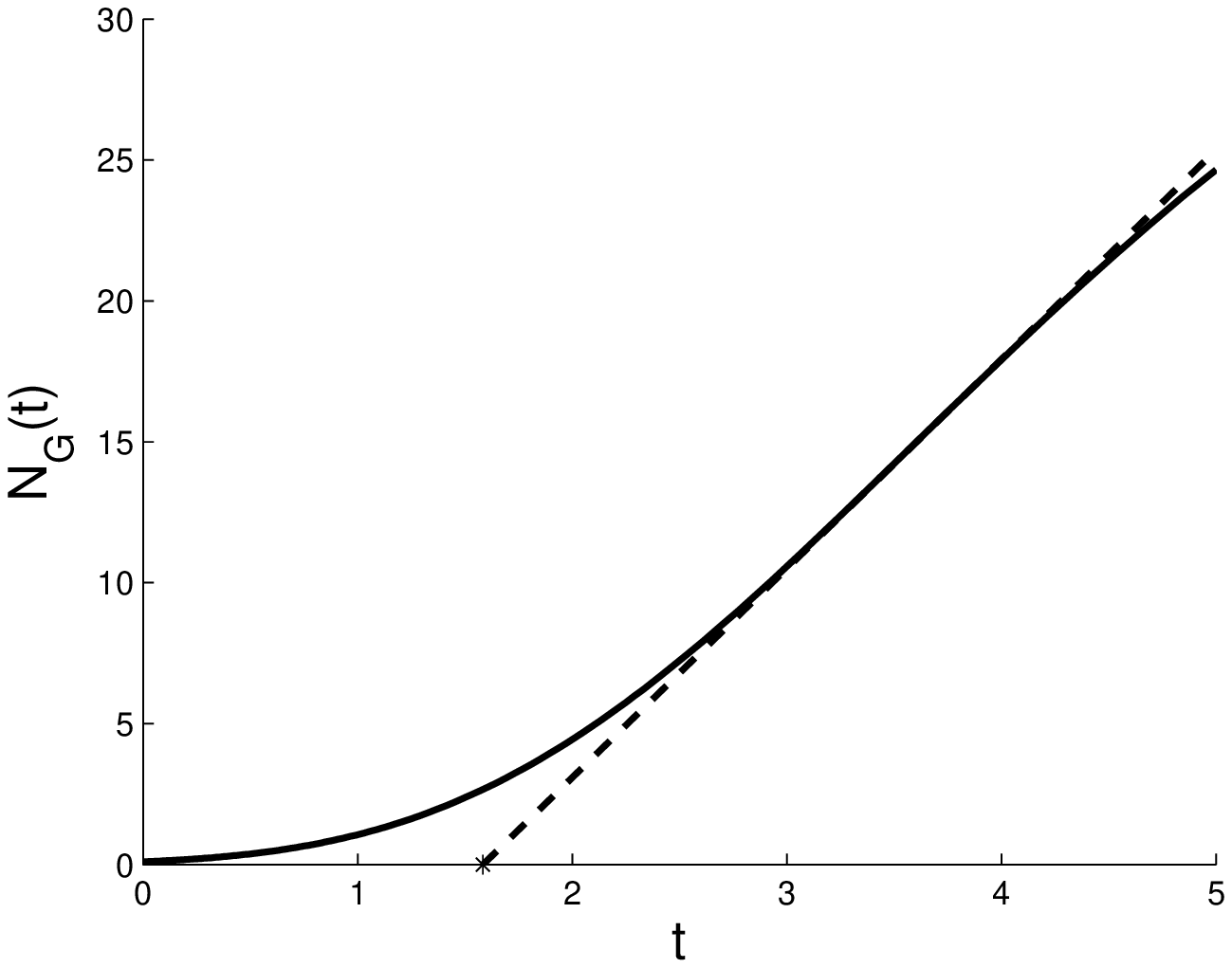}
\caption{The tangent lines (dotted curves) and the lag times (asterisk points) 
with $y=0.1$, $\alpha=3$ and $\beta=0.5$.}
\end{center}
\end{figure}
%
%
\subsection{Threshold crossing}\label{sec:Thresholdcrossing}
Studies on population growth often focus on the time that the population size spends below 
(or above) a certain threshold, say $S$. 
For instance, the upward threshold crossing problem is relevant to determine suitable therapeutic 
protocols related to Gompertz-like growth model for tumor cells evolution (cf.\ \cite{Albano_06}). 
Moreover, criteria based on threshold level crossing are also employed for intermittent cancer 
therapies (see \cite{SuzukiAihara}).  
\par
For an increasing growth model $N_\star(t)$, 
let us then analyze the time instant in which $N_\star(t)$ crosses $S$, with $S>N_\star(0)$. 
In the context of bounded populations, we consider as threshold a percentage $p$ of the 
carrying capacity $C$, so that 
$$
 S=pC,\qquad N_\star(0)<S<C.
$$ 
The (upward) threshold $S$ may represents a critical value in tumor cells dynamics or 
a superimposed boundary in  population evolution. 
Taking into account the value of $C$ in (\ref{carrying_capacity}) and that $N_\star(0)=y$, 
the following condition on $p$ holds for the model  (\ref{new_sol}):
$e^{-\alpha/\beta}<p<1$. Hence, denoting by $\theta_N$ the crossing time instant of 
$N(t)$ through the threshold $S$, and recalling (\ref{new_sol}) the solution of equation 
$N(\theta_N)=pC$ is 
$$  
 \theta_N=\left(\frac{\beta}{\alpha}\log \frac{1}{p}\right)^{-\frac{1}{\beta}}-1,
 \qquad e^{-\alpha/\beta}<p<1.
$$
Similarly, for the Korf and Gompertz models one has
$$
 \theta_K=\theta_N+1,
 \quad 0<p<1,
 \qquad 
 \theta_G={-\frac{1}{\beta}}\log\left(\frac{\beta}{\alpha}\log\frac{1}{p}\right),
 \quad e^{-\alpha/\beta}<p<1,
$$
respectively, with $\theta_G<\theta_N<\theta_K$. Such threshold crossing times are 
plotted in Fig.\ 4 as (increasing) function of $p$ in a suitable instance. 
As mentioned above, they could be employed for cancer therapies in order to establish 
suitable operational times in therapeutic protocols. 
\begin{figure}[t]
\begin{center}
\epsfxsize=7cm
\epsfbox{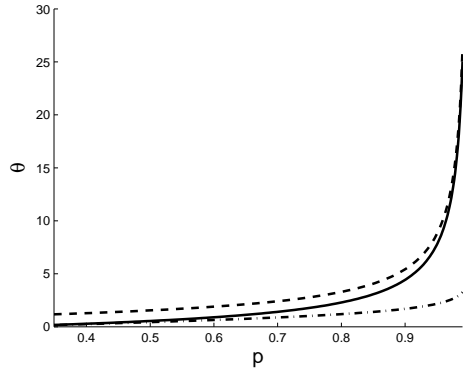}
\caption{The threshold crossing times $\theta_N$ (full line), $\theta_K$ (dashed line) and $\theta_G$ 
(dot-dashed line) as function of $p$, with $0<p<1$ for the Korf model and $e^{-\alpha/\beta}<p<1$ 
for the others, with $\alpha=2$, $\beta=1.9$ and $y=0.1$.}
\end{center}
\end{figure}
\section{Data analytic examples}\label{sect:dataex}
In this section we consider some data analytic examples for which model (\ref{new_sol}) 
provides a good fit. We deal with  the following data sets collected in Table 12.7 and Table 12.14 
of Lindsey \cite{Lindsey} in which a certain feature of the growth behavior of the small mammal 
pikas and of a colony of cells is recorded: 
\begin{itemize}
\item[(i)] 
the weight  $n_i$ of a pregnant Afghan pikas (g) over $14$ equally spaced periods from conception 
to parturition, recorded during an experiment. 
\item[(ii)] 
the size $n_i$ of a closed colony of \textit{Paramecium aurelium} are recorded at $18$ 
time instants in a given experiment.
\end{itemize}
A thorough description of the above data and the relevant experiments is provided in Lindsey \cite{Lindsey}. 
The corresponding data are fitted by nonlinear regression thanks to the classical algorithm 
based on the minimization of the sum of the squares of the differences between the measured and 
predicted values. We use the routine \textit{lsqcurvefit} of {\sc Matlab}$^{\footnotesize{\rm \textregistered}}$,     
that solves nonlinear curve-fitting (data-fitting) problems in least-squares sense, where the 
predicted values are obtained from the nonlinear equations of the models (\ref{gompertz_sol}), 
(\ref{korf_sol}) and (\ref{new_sol}). Once determined the parameters of the three models, 
in order to evaluate the attained approximations, we compare the growth curves by using 
the ISRP growth metric, introduced in Bhowmick {\em et al.}\ \cite{Bhowmick_et_al}. 
%
\begin{figure}[t]
\begin{center}
  (i) \hspace{6cm}    (ii) \\
  \epsfxsize=5.8cm
\epsfbox{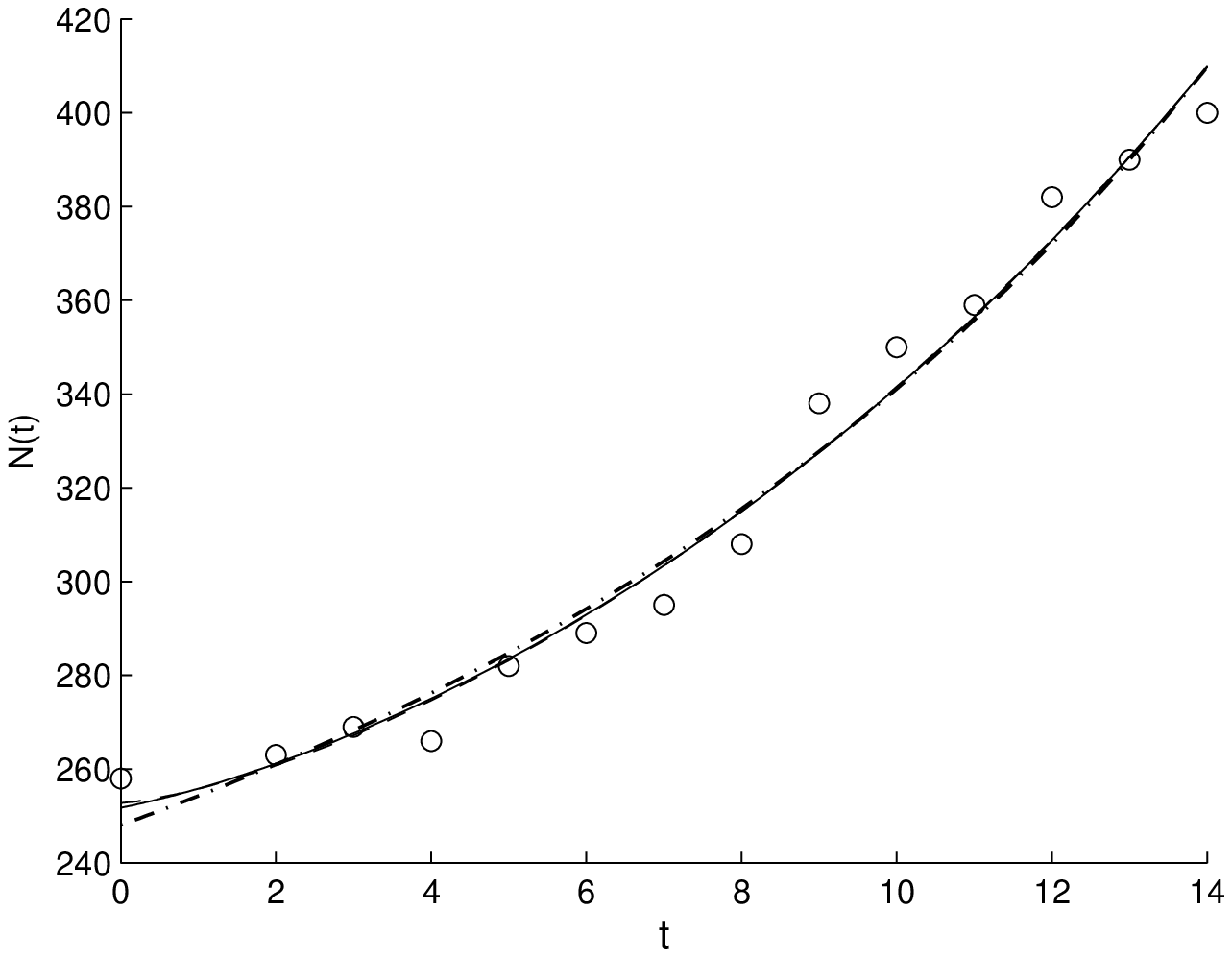}
  \epsfxsize=5.8cm
\epsfbox{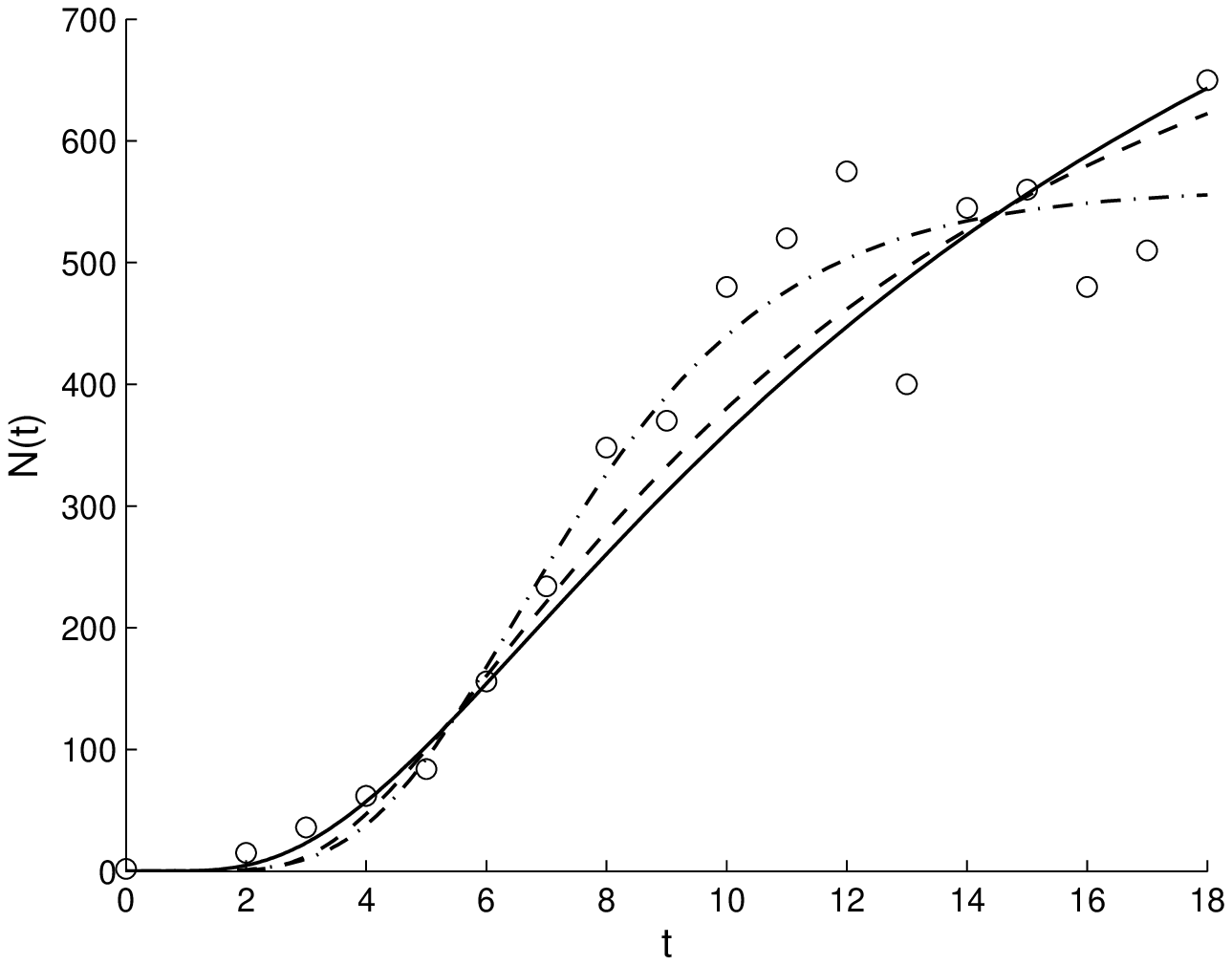}\\
\caption{Interpolation of  data sets (i), (ii) and (iii), under the growth models 
(\ref{new_sol}) (full), (\ref{korf_sol}) (dashed) and (\ref{gompertz_sol}) (dot-dashed).} 
\end{center}
\end{figure}
%
\par
For the three growth models, in Fig.\ 5 we show the interpolation of the data sets 
mentioned above. In both cases we have reasonably good fit of the data. 
We recall that the ISRP growth metric has been introduced recently in  \cite{Bhowmick_et_al} 
aiming to determine the true growth curve that best fits the data (statistically). Indeed, such a 
metric provides an estimate of the rate parameter corresponding to the identified growth model 
in specific time intervals. Note that this is in contrast to the usual $R^2$-criterion, 
which does not allow to estimate the relevant parameters. 
\par
The differential equation (\ref{new}) can be rewritten as
%
$$
 \frac{1}{N(t)}\frac{dN(t)}{dt}=\alpha(1+t)^{-(\beta+1)},\qquad t>0.
$$
%
Hence, in an experimental framework in which the population is sampled through a suitable 
time interval, according to Bhowmick {\em et al.}\ \cite{Bhowmick_et_al} parameter $\alpha$ 
can be also interpreted as the ``Overall Rate Parameter'' (when it is computed for the entire 
experimental time frame) or as the ``Interval Specific Rate Parameter'' (ISRP), denoted 
by $\alpha(\Delta t)$ (when it is estimated on the basis of a specific time interval having 
width $\Delta t$). Consequently, under the initial condition $N(0)=y$, the growth law can 
be represented in the form
\begin{equation}\label{diff_eq_metric2}
 N(t)=y  e^{\alpha s(t)}, \qquad t\geq 0,
\end{equation}
for a suitable function $s(t)$. Hence, the 
ISRP can be computed as
\begin{equation}\label{ISRP}
 ISRP= \alpha(\Delta t)=\frac{1}{s(t+\Delta t)-s(t)}\log\left(\frac{N(t+\Delta t)}{N(t)}\right).
\end{equation}
We note that for the Gompertz model, for the Korf model, and for the new proposed model, 
in all cases the growth curve has the form (\ref{diff_eq_metric2}). Hence, due to Eqs.\ 
(\ref{gompertz_sol}),  (\ref{korf_sol}) and (\ref{new_sol}), for the three models
the function $s(t)$ is given, respectively, by
$$
s_G(t)=\frac{1-e^{-\beta t}}{\beta}, 
\quad 
s_K(t)=\frac{1-t^{-\beta}}{\beta},
\quad 
s_N(t)=\frac{1-(1+t)^{-\beta}}{\beta}.
$$
It follows that the ISRP for the three models is:
%
\begin{figure}[t]
\begin{center}
  (i) \hspace{6cm}    (ii) \\
  \epsfxsize=6.2cm
\epsfbox{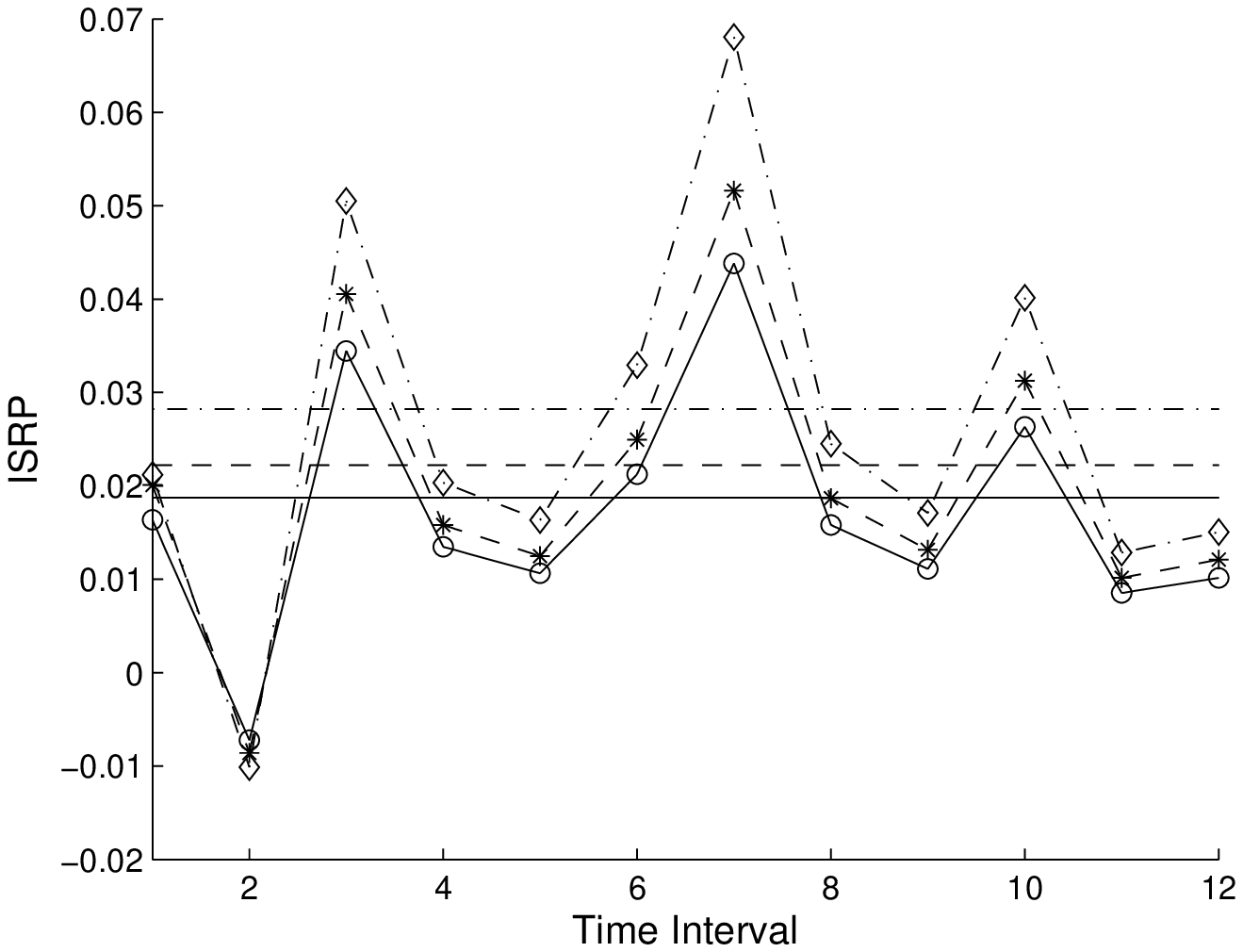}
  \epsfxsize=6.2cm
\epsfbox{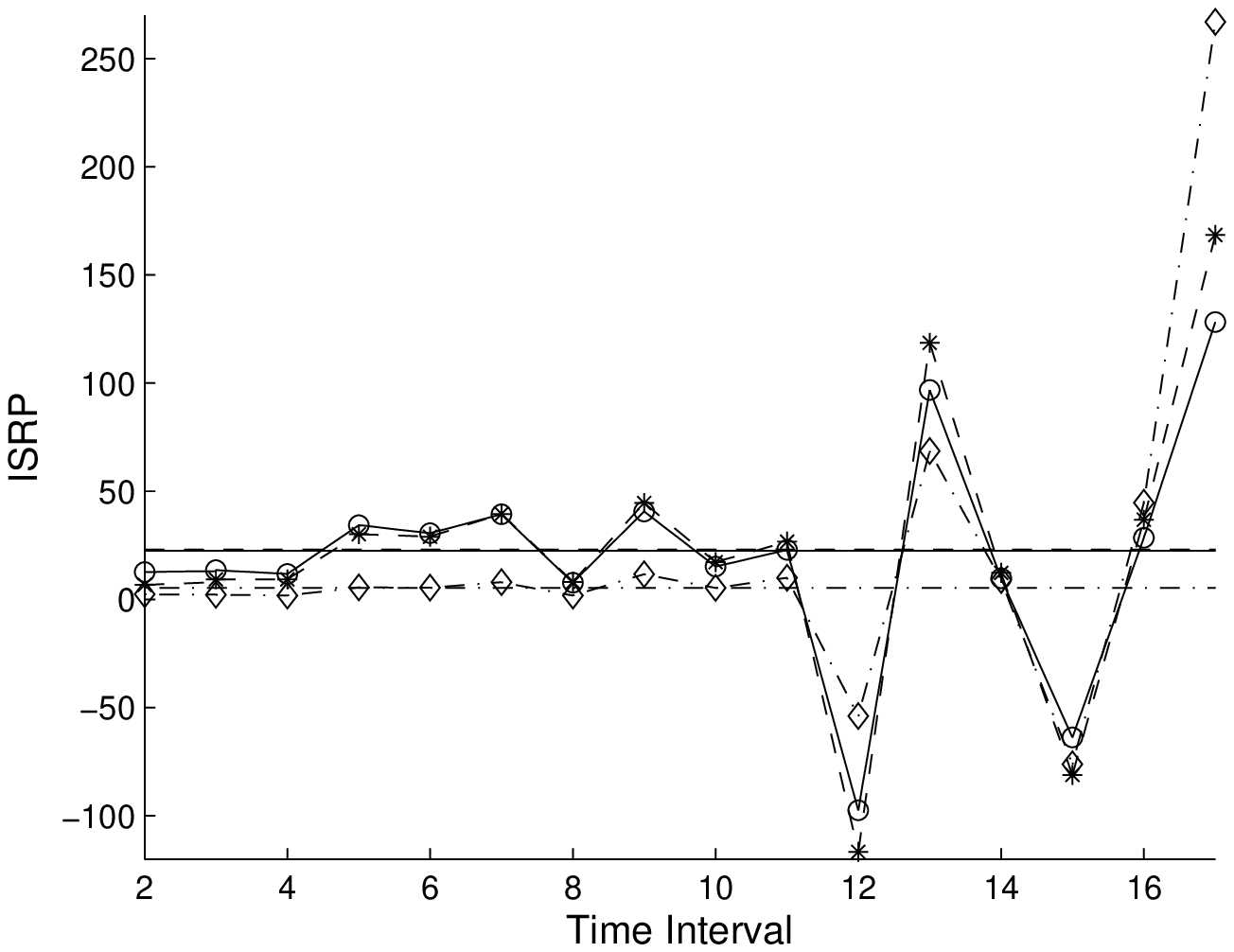} \\
\caption{ISRP and the constant estimated rate $\alpha$ are plotted for the growth models 
(\ref{new_sol}) (full), (\ref{korf_sol}) (dashed) and (\ref{gompertz_sol}) (dot-dashed).} 
\end{center}
\end{figure}
%
\begin{eqnarray}
&& ISRP_G= \frac{\beta}{e^{-\beta t}-e^{-\beta (t+\Delta t)}}
\log\left(\frac{N(t+\Delta t)}{N(t)}\right), 
\nonumber\\
&& ISRP_K= \frac{\beta}{t^{-\beta}-(t+\Delta t)^{-\beta}}
\log\left(\frac{N(t+\Delta t)}{N(t)}\right),
\nonumber \\
&& ISRP_N= \frac{\beta}{(1+t)^{-\beta}-(1+t+\Delta t)^{-\beta}}
\log\left(\frac{N(t+\Delta t)}{N(t)}\right).
\nonumber
\end{eqnarray}
For a more tight comparison, the 3 models are then used to fit the datasets of cases 
(i) and (ii). Their comparative performances are discussed based on the ISRP metric. 
The extent of departure between the line of the constant rate parameter $\alpha$ of the 
corresponding growth law and the estimated ISRP reflects the deviation of the assumed 
model from the true law. The ISRP and the constant estimated rate $\alpha$ are plotted 
in Fig.\ 6 for the 3 models, in order to rank such models by means of the ISRP metric.
We have the following estimated values of $\alpha$ in the two considered instances:
\begin{description}
\item{(i)}  $\alpha_N=0.0187$, \quad $\alpha_K =0.0222$, \quad $\alpha_G =0.0282$,
\item{(ii)}  $\alpha_N=22.5261$, \quad $\alpha_K =23.0986$, \quad $\alpha_G =5.3504$.
\end{description}
Specifically, the proposed model (\ref{new_sol}) has the smallest average deviation 
of ISRP from the estimated parameter. Indeed, by computing the $d_2$-distance  (in the 
euclidean norm) between ISRP and the constant parameter,  for the given datasets, one has 
\begin{description}
\item{(i)}  $d_2^N=0.0442$, \quad $d_2^K =0.0522$, \quad $d_2^G =0.0672$,
\item{(ii)}  $d_2^N=200.2610$, \quad$d_2^K =250.3044$, \quad $d_2^G=290.2033$.
\end{description}
It thus follows that the smallest distance is obtained for the proposed model (\ref{new_sol}). 
Hence, the actual dynamics of the two datasets is better described by model (\ref{new_sol}) 
rather than the other two models.
\section{Analysis of a special inhomogeneous linear birth-death process}\label{sect:linearBD}
Deterministic curves are often employed to describe population growth due to their ease of  
tractability. However they can be reasonably used just to describe overall dynamics. A more 
realistic description can be performed by means of stochastic models. Indeed, suitable random 
processes allow to take into account the ubiquitous environmental fluctuations, which are due 
to various factors that are unknown or not always quantifiable. As example, we recall the 
paper by Tan \cite{Tan}, that investigates a birth-death process whose mean is identical to 
the curve of the Gompertz growth. A general theory based on stochastic differential equations 
is proposed by Tan et al.\ \cite{Tan_2001} to construct models for carcinogenesis and to 
analyze related data. A non-homogeneous density-dependent birth-death process 
describing a stochastic logistic growth has been investigated by 
Tan and Piantadosi \cite{Tan_Piantadosi}. More recently, aiming to describe 
the effect of radiation therapy, Tuckwell \cite{Tuckwell} considered the differential equation 
of Gompertzian growth modified by inclusion of a stochastic term involving the Poisson process. 
\par
In order to propose a stochastic counterpart of the growth model introduced in (\ref{new_sol}),  
in this section we shall investigate an evolutionary model based on a birth-death process. 
Since a population described by the curve (\ref{new_sol}) can reach a great size when 
$\alpha$ is large, we are driven to consider a stochastic process having infinite state-space. 
Specifically, 
we assume that the number of individuals of a population is described by an inhomogeneous 
linear birth-death process $\left\{X(t);\,t \geq 0\right\}$ having state space $\mathbb{N}_0$, 
with 0 absorbing endpoint. The birth and death rates are specified by 
\begin{equation}
\begin{array}{ll}
 P[X(t+h)=j+1\,|\,X(t)=j] = j \lambda(t) h +o(h), & \qquad j\in \mathbb{N}_0 \\[1mm]
 P[X(t+h)=j-1\,|\,X(t)=j] = j \mu(t) h+o(h), & \qquad j \in \mathbb{N}, 
\end{array}
\label{def_proc_prob}
\end{equation}
where $\lambda(t)$ and $\mu(t)$ are positive functions, integrable on $(0,t)$ for any finite $t> 0$, and $h>0$. 
Since the rates (\ref{def_proc_prob}) are linear in $j$, new births and deaths occur proportionally to the 
population size, and according to the time-dependent rates $\lambda(t)$ and $\mu(t)$, which constitute 
the individual birth rate and death rate at time $t$, respectively. 
We recall that the model specified in (\ref{def_proc_prob}) is a time-inhomogeneous 
version of the classical linear birth-death process (see, for instance, Bailey \cite{Bailey});  
examples of applications of such a model in genetics and phylogenetics
have been reviewed in Novozhilov {\em et al.}\ \cite{NoKaKo2006}. 
\par
For all $t\geq 0$, and $x,y\in \mathbb{N}_0$ we denote the transition probability of $X(t)$ by 
$P_{y,x}(t)=P[X(t)=x\,|\,X(0)=y]$. 
Clearly, $P_{y,x}(t)$ represents the probability that the population size  equals level $x$ at time $t$, 
conditional on the initial size $y$. 
Formally, the initial state $y$ may be 0, however this choice is trivial since 0 is an absorbing endpoint 
for the birth-death process. Henceforth we thus assume $X(0)=y\in \mathbb{N}$, in agreement with 
the fact that the initial state of the growth model (\ref{new_sol}) is positive. 
\par
Aiming to determine the transition probability, 
for $0< z< 1$ and $t\geq 0$, let $G(z,t)=\sum_{x=0}^{\infty}P_{y,x}(t)z^x$ 
be the probability generating function of $X(t)$, with initial condition $G(z,0)=z^y$. 
As showed in Tan \cite{Tan}, it results
\begin{equation}\label{prob_gen_fun}
G(z,t)=\left\{1-(z-1)[(z-1)\phi(t)-\psi(t)]^{-1}\right\}^y,
\end{equation}
where 
\begin{equation}\label{fun_phi_psi}
 \psi(t)=\exp{\left\{-\int_0^t [\lambda(\tau)-\mu(\tau)]d\tau\right\}}, 
 \qquad 
 \phi(t)=\int_0^t \lambda(\tau) \psi(\tau) d\tau.
\end{equation}
For brevity, in the following we denote the conditional mean and the conditional variance 
of $X(t)$, given $X(0)=y$, by 
$$
 E_y(t)=E[X(t)|X(0)=y]\quad \hbox{and} \quad V_y(t)=Var[X(t)|X(0)=y],
$$ 
respectively. Clearly, $E_y(t)$ and $V_y(t)$ describe the mean trend and the related 
variability in the stochastic growth model.  From assumptions (\ref{def_proc_prob}) we obtain 
that the population mean satisfies the following differential equation: 
\begin{equation}
 \frac{d E_y(t)}{dt}=\xi(t) E_y(t),\qquad t>0,
\label{eq:eqdiffEyt}
\end{equation}
where  the time-dependent growth rate $\xi(t)$ 
is the net growth rate {\em per capita} of individuals, i.e.
\begin{equation}
 \xi(t)=\lambda(t)- \mu(t), \qquad t\geq 0.
\label{eq:xidifflambdamu}
\end{equation}
The behavior of $X(t)$ may mimic the growth curve proposed in (\ref{new_sol}). Indeed, 
the analogies between the growth model (\ref{new_sol}) and the considered birth-death process 
stems from the fact that Eq.\ (\ref{eq:eqdiffEyt}) has the same form of (\ref{eq:nt}). 
Hence, the growth model (\ref{new_sol}) and the mean of the birth-death process $X(t)$ are 
governed by the same equation in the special case when the assumptions (\ref{eq:xidit}) 
and (\ref{eq:xidifflambdamu}) hold. This case will be treated in Section \ref{sect:spcase}. 
\par
Let us now obtain the transition probabilities, the conditional mean and the 
conditional variance of $X(t)$.
\begin{proposition}\label{prop:1}
Let $y\in \mathbb{N}$. 
The transition probabilities of the process $X(t)$, with rates (\ref{def_proc_prob}), are given by:
\begin{eqnarray}\label{prob_j}
 &&  P_{y,0}(t)=\left(1-\frac{1}{\psi+\phi}\right)^y, \\
 &&  P_{y,x}(t)=\left(\frac{\phi}{\psi+\phi}\right)^x \,\sum_{i=0}^{m} {y \choose u}{y+x-i-1 \choose y-1}
 (\phi^{-1}-1)^{i}\left(1-\frac{1}{\psi+\phi}\right)^{y-i}\nonumber \\
&& \hspace{10cm} x\in\mathbb{N},\nonumber
\end{eqnarray}
with $m=\min\left\{y,x\right\}$, where $\psi=\psi(t)$ and $\phi=\phi(t)$ are shown 
in (\ref{fun_phi_psi}). Moreover, the conditional mean and the conditional variance of $X(t)$ are
\begin{equation}\label{cond_mean_var}
E_y(t)=\frac{y}{\psi(t)},
\qquad 
V_y(t)=y\,\frac{[\psi(t)+2 \phi(t)-1]}{\psi^2(t)},
\qquad t\geq 0,
\end{equation}
respectively.  
\end{proposition}
\begin{proof}
The probability generating function given in (\ref{prob_gen_fun}) can be rewritten as 
the product of two generating functions, i.e.
$$
G(z,t)=  \left( {1-z \; \frac{\phi}{\phi+\psi}} \right)^{-y}  
 \left( {z\; \frac{1-\phi}{\phi+\psi}+1-\frac{1}{\phi+\psi}}\right)^y.
$$
Since $0<\frac{\phi}{\phi+\psi}<1$, for $0<z<1$ one has
$$
 \left(1-z\;\frac{\phi}{\phi+\psi}\right)^{-y}
 =\sum_{i=0}^\infty a_i z^i, 
 \qquad 
 \left(z\; \frac{1-\phi}{\phi+\psi}+1-\frac{1}{\phi+\psi}\right)^y
 =\sum_{i=0}^\infty b_i z^i, 
$$ 
where 
$$
 a_i={y+i-1 \choose y-1} \left(\frac{\phi}{\phi+\psi}\right)^i, 
 \qquad 
 i\in \mathbb{N}_0
$$
and 
$$ 
 b_i={y \choose i} \left(\frac{1-\phi}{\phi+\psi}\right)^i \left(1-\frac{1}{\phi+\psi}\right)^{y-i} 
 {\bf 1}_{\{0\leq i\leq y\}}.
$$ 
Therefore, we have $P_{y,x}(t)=\sum_{i=0}^x b_i a_{x-i}$ and the expressions (\ref{prob_j}) 
thus follow. Finally, the conditional mean and the conditional variance (\ref{cond_mean_var}) 
can be obtained from (\ref{prob_gen_fun}) straightforwardly. 
\end{proof}
\par
It is worth pointing out that the expressions (\ref{prob_j}) actually correct the transition 
probabilities given in Tan \cite{Tan}.  
\par
The monotonicity and the concavity of the conditional mean and variance of $X(t)$ can be 
analyzed in term of $\lambda(t)$ and $\mu(t)$ by noting that, due to (\ref{cond_mean_var}), 
for $t\geq 0$ we have   
$$
 E_y'(t)=y\frac{\xi(t)}{\psi(t)},
 \qquad 
 E_y''(t)=y\frac{\lambda'(t)-\mu'(t)+\xi^2(t)}{\psi(t)}, 
$$
$$
 V_y'(t)=\frac{y}{\psi^2(t)}\left\{\lambda(t)[3\psi(t)+4\phi(t)-2]-\mu(t)[\psi(t)+4\phi(t)-2]\right\},
$$
\begin{eqnarray}
 V''_y(t)&=& \frac{y}{\psi^2(t)} \left\{\lambda'(t)\left[3 \psi(t)+4\phi(t)-2\right]-\mu'(t)\left[\psi(t)+4\phi(t)-2\right]\right.
 \nonumber\\
&&+\lambda(t)\left[\xi(t)(3\psi(t)+8\phi(t)-4)+4\lambda(t)\psi(t)\right]
\nonumber\\
&&\left.-\mu(t)\left[\xi(t)(\psi(t)+8\phi(t)-4)+4\lambda(t)\psi(t)\right]\right\},
 \nonumber
\end{eqnarray}
where $\xi(t)$ is the net growth rate defined in (\ref{eq:xidifflambdamu}). It immediately follows that the 
population mean is increasing (decreasing) at time $t$ if $\xi(t)$ is positive (negative).  The concavity of the 
mean and the monotonicity and  concavity  of the variance can be studied through suitable computations.  
From the above expressions, some results on the conditional mean and  variance 
are provided in Table 4, where ${\cal D}$ means that the result strictly depends on the 
rates $\lambda(t)$ and $\mu(t)$, and where we have set
\begin{equation}\label{setting_tilde}
\tilde{\psi}=\lim_{t \rightarrow \infty} \psi(t),
\quad 
\tilde{\phi}=\lim_{t \rightarrow \infty} \phi(t),
\quad
\tilde{\lambda}=\lim_{t \rightarrow \infty} \int_0^t \lambda(\tau) d\tau,
\quad
\tilde{\mu}=\lim_{t \rightarrow \infty} \int_0^t \mu(\tau) d\tau.
\end{equation}
It is shown that the present model can exhibit very different behaviors, according to the values of the relevant 
parameters. For instance, the conditional mean  may increase toward $\infty$ and decrease to a constant 
(possibly vanishing) asymptotic value. 
%
\begin{table}[t]
\begin{center}
\begin{tabular}{lllll}
\hline
\small
$\!\!\!\!\xi(t)$ &   $E_y(t)$ &  
$\displaystyle\lim_{t \rightarrow \infty} E_y(t)$ &   
$V_y(t)$ &  $\displaystyle\lim_{t \rightarrow \infty} V_y(t)$ 
\\
\hline
$\!\!\!\!>0\;\;\forall t>0$ &  
$\begin{array}{l}
  \hbox{strictly}    \\
  \hbox{increas.}
\end{array}$
& %
$
\left\{
\begin{array}{ll}
y/\tilde{\psi}, & \hbox{if }\tilde{\lambda}-\tilde{\mu}<\infty \\
\infty, & \hbox{if }\tilde{\lambda}-\tilde{\mu}=\infty
\end{array} 
\right.
$
                     & 
$\begin{array}{l}
  \hbox{strictly}    \\
  \hbox{increas.}
\end{array}$      &   $
\left\{
\begin{array}{l}
y(\tilde{\psi}+2 \tilde{\phi}-1)/\tilde{\psi}^2,   \\
\qquad\quad \hbox{if }\tilde{\lambda}<\infty, \tilde{\mu}<\infty   \\
\infty, \qquad  \hbox{if }\tilde{\lambda}=\infty   
\end{array} 
\right.
$
\\
        &      &      &   &     \\
\hline
$\!\!\!\!=0\;\;\forall t>0$ &   constant &  $y$  &  
$\begin{array}{l}
  \hbox{strictly}    \\
  \hbox{increas.}
\end{array}$       
& $
\left\{
\begin{array}{ll}
2y\tilde{\lambda}, & \hbox{if }\tilde{\lambda}<\infty\\
\infty, & \hbox{if }\tilde{\lambda}=\infty
\end{array} 
\right.
$ \\
         &            &     &    &     \\
\hline
$\!\!\!\!<0\;\;\forall t>0$  &  
$\begin{array}{l}
  \hbox{strictly}    \\
  \hbox{decreas.}
\end{array}$   
& $
\left\{
\begin{array}{ll}
y/\tilde{\psi}, & \hbox{if }\tilde{\mu}-\tilde{\lambda}<\infty\\  
0, & \hbox{if }\tilde{\mu}-\tilde{\lambda}=\infty\\
\end{array} 
\right.
$   &   ${\cal D}$  &  $
\left\{
\begin{array}{l}
y(\tilde{\psi}+2 \tilde{\phi}-1)/\tilde{\psi}^2, \\
\qquad \quad \hbox{if }\tilde{\lambda}<\infty, \tilde{\mu}<\infty \\  
{\cal D},  \qquad \hbox{otherwise}  
\end{array} 
\right.
$            \\
        &   &    &     &\\
\hline
\end{tabular}
\end{center}
\caption{Some results on the conditional mean and  variance of $X(t)$.}
\end{table} 
\par
We finally remark that, since 0 is an absorbing endpoint, $P_{y,0}(t)$ represents 
the probability that the population reaches extinction prior to time $t$. 
Due to (\ref{prob_j}), the probability of ultimate extinction, for initial state $y$, is given by 
\begin{equation}\label{limPt}
 \pi_{y,0}:=\lim_{t\to \infty} P_{y,0}(t)
 =\left(1-\frac{1}{\tilde\psi+\tilde\phi}\right)^y, 
 \qquad y\in \mathbb{N}.
\end{equation}
Clearly,  if $\tilde{\mu}-\tilde{\lambda}=\infty$ or $\tilde{\phi}=\infty$, then $\pi_{y,0}=1$, 
this leading to certain ultimate extinction. 
\subsection{Analysis of a special case}\label{sect:spcase}
We already pointed out that the growth model (\ref{new_sol}) and the 
conditional mean of the birth-death process $X(t)$ satisfy the same linear differential 
equation, when both Eqs.\ (\ref{eq:xidit}) and (\ref{eq:xidifflambdamu}) hold. 
The aim of this section is to investigate in detail such a special case. 
\begin{proposition}\label{propEl}
The linear birth-death process with rates specified in (\ref{def_proc_prob}) has 
conditional mean 
\begin{equation}\label{cond_exp}
 E_y(t)= y \exp{\left\{\frac{\alpha}{\beta}\left[1-(1+t)^{-\beta}\right]\right\}}, 
 \qquad t\geq 0
\end{equation}
if and only if 
\begin{equation}\label{rel_lam_mu}
\lambda(t)-\mu(t)=\alpha (1+t)^{-(\beta+1)}, \qquad t\geq 0.
\end{equation}
\end{proposition}
\begin{proof}
From the first of (\ref{fun_phi_psi}) and (\ref{cond_mean_var}) we have 
$$
 E_y(t)=y \exp{\left\{\int_0^t [\lambda(\tau)-\mu(\tau)]d\tau\right\}}, \qquad t\geq 0.
$$ 
Hence, the expression (\ref{cond_exp}) holds if and only if 
\begin{equation}
 \int_0^t [\lambda(\tau)-\mu(\tau)]d\tau=\frac{\alpha}{\beta}\left[1-(1+t)^{-\beta}\right],
 \qquad t\geq 0,
 \label{rel_int_lam_mu}
\end{equation}
that is equivalent to (\ref{rel_lam_mu}). 
\end{proof}
\par
Let us now investigate the process $X(t)$ under the assumption (\ref{rel_lam_mu}). In this case, 
the birth rate is larger than the death rate, and the net growth rate defined in 
(\ref{eq:xidifflambdamu}) is decreasing and tends to zero according to a power law. 
Hence, this assumption leads to results useful to describe populations that experience a tendential high growth 
for initial time, which however decreases in time. Furthermore, by virtue of 
(\ref{fun_phi_psi}), the mean $E_y(t)$ is identical to the curve $N(t)$ given in 
(\ref{new_sol}). Hence, due to (\ref{rel_lam_mu}), from (\ref{rel_int_lam_mu}) one has 
$\tilde{\lambda}-\tilde{\mu}={\alpha}/{\beta}$, with $\tilde{\lambda}$ and $\tilde{\mu}$ 
defined in (\ref{setting_tilde}). From the results shown in Table 4, $E_y(t)$ is 
strictly increasing, and clearly it tends to the carrying capacity, i.e.\ 
$\displaystyle\lim_{t\to \infty} E_y(t)=y e^{\alpha/\beta} \equiv C$. Moreover, the function $\psi(t)$ 
is given by
\begin{equation}\label{psi}
\psi(t)=\exp{\left\{-\frac{\alpha}{\beta}\left[1-\left(1+t\right)^{-\beta}\right]\right\}}, 
\qquad t\geq 0,
\end{equation}
with $\tilde\psi=e^{-\alpha/\beta}$, so that the variance $V_y(t)$ is strictly increasing in 
$t$, according to the results specified in Table 4. 
\begin{example}\rm 
Let Eq.\ (\ref{rel_lam_mu}) be satisfied. 
We now consider various instances of $\mu(t)$ listed in Table 5. 
The corresponding transition probabilities $P_{y,x}(t)$, given in (\ref{prob_j}), are plotted in 
Fig.\ 7 for $y=1$, $x=0$ and $x=1$, and for some choices of the parameters. The asymptotic 
absorption probability (\ref{limPt}) is provided in the last column of Table 5. Note that in the first 
three cases of Table 5 we have $\tilde \mu=+\infty$, and thus $\tilde \psi=+\infty$, this implying 
that $\pi_{y,0}=1$, i.e.\ the ultimate extinction is certain when the individual death rate $\mu(t)$ 
is constant, or is a ramp function, or is sinusoidal. On the contrary, 
in case (d) it is $\tilde \mu<+\infty$, and thus $\tilde \psi<+\infty$, so that $\pi_{y,0}<1$. 
Specifically, if the individual birth and death rates are proportional and follow a decreasing power-law rule, 
then the ultimate extinction is not certain and is decreasing in $y$ (the initial population size). 
%
\begin{table}[t]
\begin{center}
\begin{tabular}{clc}
\hline
   case        &   $\mu(t)$ &  $\pi_{y,0}$\\[1mm]
\hline
(a)       &       $A$  &  $1$                                          \\[1mm]
\hline
(b) 	   &         $A+B \,(t-t_0){\bf 1}_{ \{t\geq t_0\} }$    &  $1$                         \\[1mm]   
\hline
(c)        &    $A+B\, \sin\left(\frac{2\pi }{Q}t\right)$      &  $1$             \\[1mm]
\hline
(d)    &      $A \,(1+t)^{-(\beta+1)}$    &  $\left[\frac{\frac{A}{\alpha}
\left(1-e^{-\frac{\alpha}{\beta}}\right)}{\frac{A}{\alpha}\left(1-e^{-\frac{\alpha}{\beta}}\right)+1}\right]^y$              \\
\hline
\end{tabular}
\end{center}
\caption{Some choices of $\mu(t)$ and the corresponding extinction probabilities (\ref{limPt}). 
Parameters $A$, $B$,  $Q$ and $t_0$ are positive constants.}
\end{table} 
%
%
\begin{figure}[t]
\begin{center}
\epsfxsize=5.8cm
\epsfbox{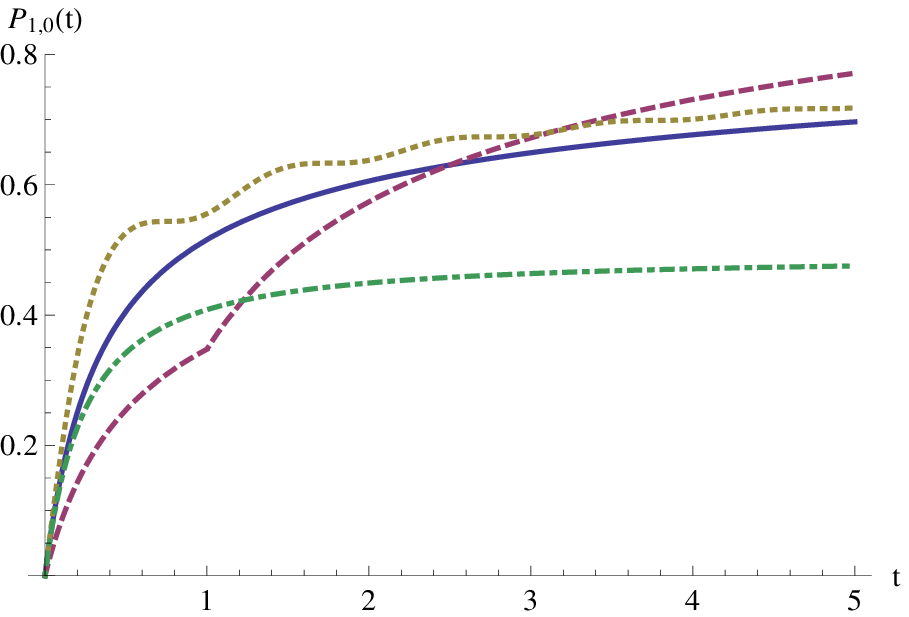}
\epsfxsize=5.8cm
\epsfbox{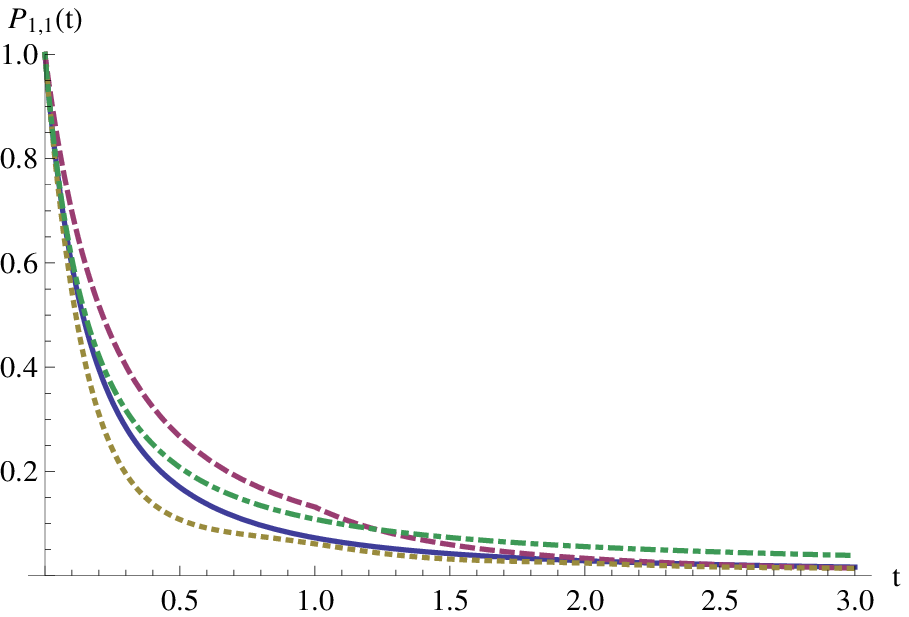}
\caption{Probabilities $P_{1,0}(t)$ and $P_{1,1}(t)$ for the cases specified in Table 5, with  
(a) $A=2.00001$ (full line), (b) $B=1$, $A=2.00001$,  $ t_0=1$ (dashed line), 
(c) $A=2.00001$, $B=2$, $Q=1$ (dotted line), (d) $A=2.00001$ (dot-dashed line).}
\end{center}
\end{figure}
%
\begin{figure}[t]
\begin{center}
\epsfxsize=5.8cm
\epsfbox{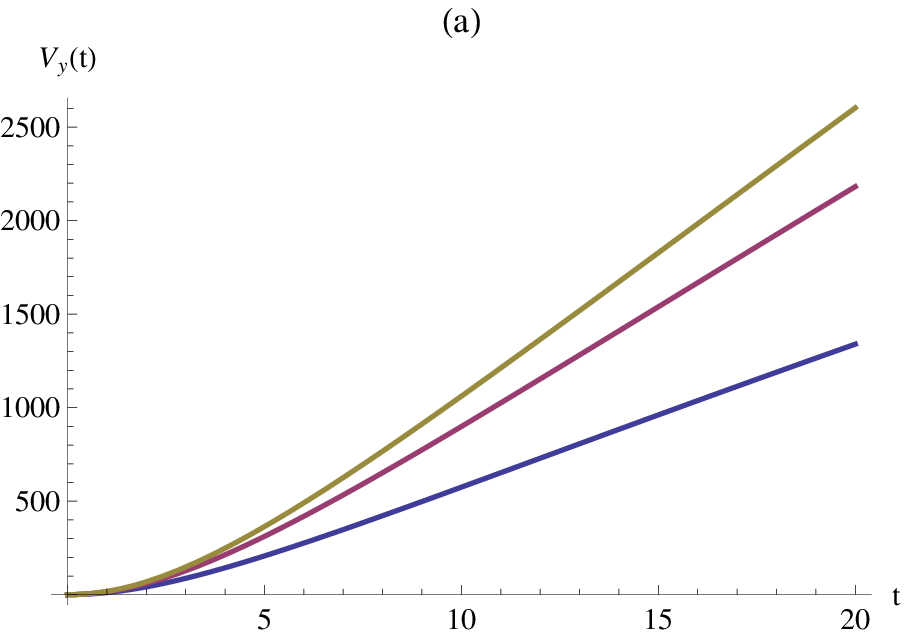}
\epsfxsize=5.8cm
\epsfbox{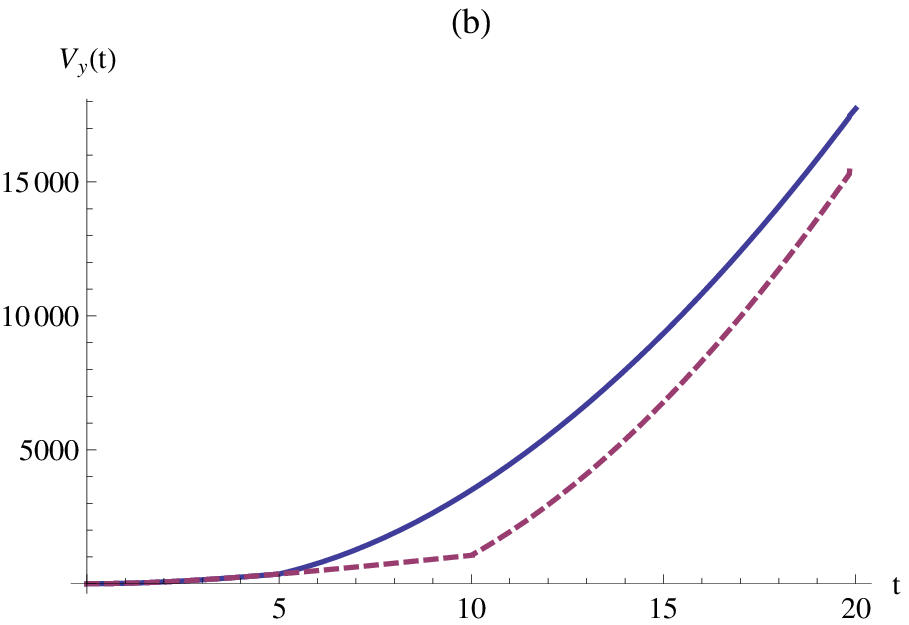}\\
\epsfxsize=5.8cm
\epsfbox{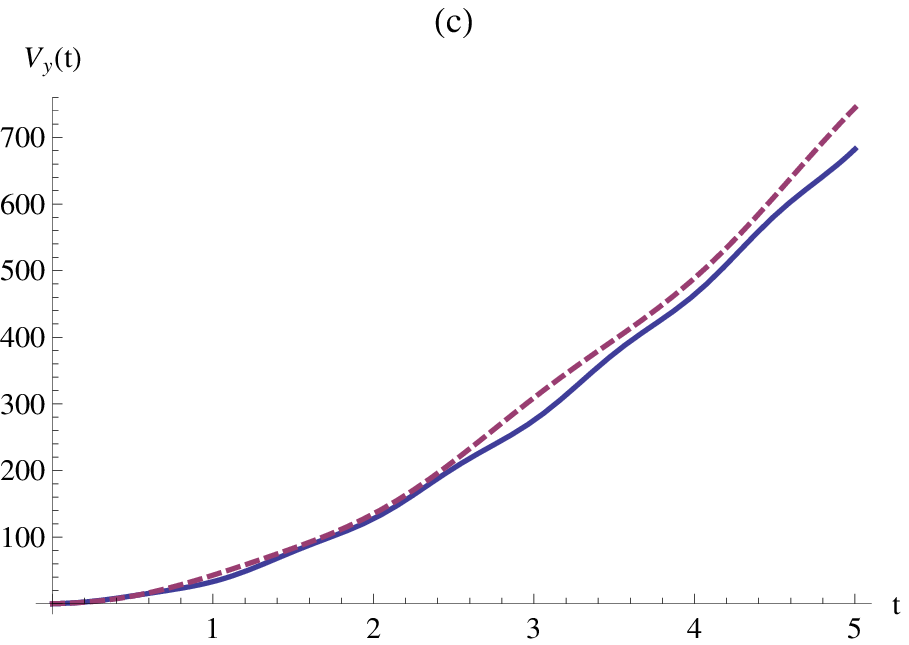}
\epsfxsize=5.8cm
\epsfbox{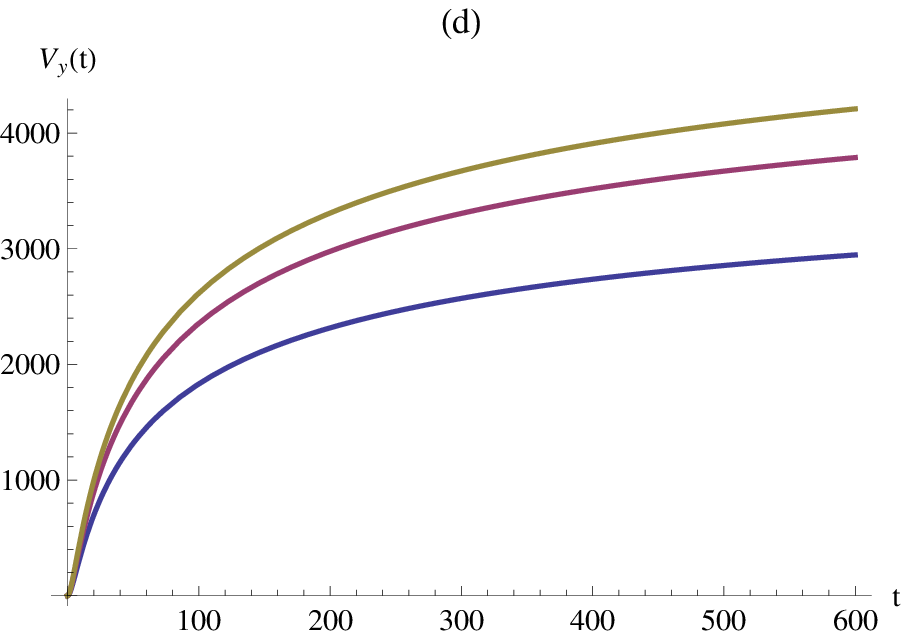}
\caption{For $y=1,\,\alpha=2,\, \beta=0.5$, the variance $V_y(t)$ is plotted for the cases of Table 5. 
In case (a): $A = 0.4,\, 0.8,\,1$ (from bottom to top). 
In case (b): $A=1$, $B=1$, $t_0=5$ (solid line) and $t_0=10$ (dashed line). 
In case (c): $A=2.00001$, $B=2$, $Q=1$ (solid line) and $Q=2$ (dashed line). 
In case (d): $A=0.4,\,0.8,\,1$ (from bottom to top), where the limits of $V_y(t)$ as $t\to \infty$ 
are $4096.9, 5267.45, 5852.72$, respectively.}
\end{center}
\end{figure}
\par
We now study the variance $V_y(t)$ when Eq.\ (\ref{rel_lam_mu}) holds, for the choices 
of $\mu(t)$  given in Table 5, with $\psi(t)$ specified in (\ref{psi}). 
\begin{itemize}
\item[(a)] 
Let us assume that $\mu(t) = A$, with $A>0$. This case refers to populations that experience 
constant individual death rate, regardless of age. Due to (\ref{cond_mean_var}) the variance is 
$$
V_y(t)=\frac{y}{\psi^2(t)} \left\{1-\psi(t)-\frac{2A}{\beta}
 \left(-\frac{\alpha}{\beta}\right)^{\frac{1}{\beta}}e^{-\frac{\alpha}{\beta}}\gamma \right\},
\qquad t\geq 0,
$$
with $\gamma=\Gamma\left(-\frac{1}{\beta},-\frac{\alpha}{\beta}\right)
-\Gamma\left(-\frac{1}{\beta},-\frac{\alpha}{\beta}(1+t)^{-\beta}\right)$, 
where $\Gamma(\cdot,\cdot)$ is the upper incomplete Gamma function. 
\item[(b)] We consider $\mu(t)=A+B\,(t-t_0) {\bf 1}_{ \{t\geq t_0\} }$, with $A>0$,  $B>0$ and $t_0>0$. 
In this case the individual death rate is constant until time $t_0$ and is linear increasing afterward, 
for instance due to worsening of the environmental or individual conditions. 
The variance has a rather cumbersome form, and thus it is omitted for brevity. 
\item[(c)] Let $\mu(t)=A+B \sin\left(\frac{2 \pi }{Q}t \right)$, where $Q>0$ and $A > |B|>0$. In this case 
the individual death rate is sinusoidal, thus describing populations that are subject to periodic 
increase and decrease of mortality, for instance due to seasonal predation or environmental variability. 
The variance  can be expressed in integral form. Again, we omit the result for brevity. 
\item[(d)] Let $\mu(t)= A (1+t)^{-(\beta+1)} $, with $A>0$, 
so that the individual birth and death rates are proportional and follow a decreasing power-law rule; 
from (\ref{cond_mean_var}) one obtains the variance    
$$
V_y(t)=\frac{y\left(1+2\frac{A}{\alpha}\right)
 \left(1-e^{-\frac{\alpha}{\beta}[1-(1+t)^{-\beta}]}\right)}{e^{-\frac{2\alpha}{\beta}[1-(1+t)^{-\beta}]}},
 \qquad t\geq 0.
$$
\end{itemize}
Note that in the first three cases the variance diverges as $t\to \infty$, whereas in case (d) 
it tends to a constant. Some plots of $V_y(t)$ are shown in Fig.\ 8, 
showing that the shape of the variance depends on $\mu(t)$.
\end{example}	
\par 
We point out that even though the birth-death process considered in this section has 
conditional mean identical to the growth curve (\ref{new_sol}), 
the behavior of the sample paths of $X(t)$ may be significantly different from such growth 
curve, since they can be absorbed at zero (recall Eqs.\ (\ref{prob_j}) and (\ref{limPt})). 
Although this characteristic allows $X(t)$ to be more realistic for some biological applications, 
in order to make the stochastic process closer to the growth curve, hereafter we remove the possibility 
of downward jumps. This leads to a birth process, that is clearly more appropriate to describe pure 
growth phenomena since its sample paths are non-decreasing. 
\section{Analysis of a special time-inhomogeneous linear birth process} \label{sect:simpleB}
In this section we investigate a time-inhomogeneous birth process, which is a special 
case of the process studied in Section \ref{sect:spcase}, whose mean is identical to the growth 
curve proposed in Eq.\ (\ref{new_sol}). 
\par
Assume that $\mu(t)\equiv 0$ in Eq.\ (\ref{def_proc_prob}), so that the number of individuals of a 
population is now modeled by an inhomogeneous linear birth process $\left\{X(t);\,t \geq 0\right\}$. 
Let ${\cal S}=\{y,y+1,\ldots\}$ be the state space of $X(t)$, where $y\in \mathbb{N}$ is the initial state, and let  
\begin{eqnarray}\label{prob_j_simple}
 && P[X(t+h)=j+1|X(t)=j]=j \lambda(t) h+o(h), \qquad j\in {\cal S}\\
 && P[X(t+h)=j|X(t)=j]=1-j \lambda(t)h+o(h), \qquad j\in {\cal S}, 
 \nonumber
\end{eqnarray}
where $\lambda(t)$ is a continuous positive function, integrable on $(0,t)$ for any finite $t>0$. 
Again, let $P_{y,x}(t)=P[X(t)=x\,|\,X(0)=y]$ be the transition probabilities of $X(t)$. It is well 
known (see \cite{Ricciardi}, for instance) that for $y\in \mathbb{N}$ and $x\in {\cal S}$ we have 
\begin{equation}\label{trans_prob_simple}
P_{y,x}(t)={x-1\choose y-1} e^{-y \Lambda(t)}\left(1-e^{-\Lambda(t)}\right)^{x-y},\qquad t\geq 0,
\end{equation}
where 
\begin{equation}\label{integral}
 \Lambda(t)=\int_0^t \lambda(s)ds, 
 \qquad t\geq 0. 
\end{equation}
In agreement with Proposition \ref{propEl}, the following result holds. 
\begin{proposition}
The linear birth process with rates specified in (\ref{prob_j_simple}) has 
conditional mean (\ref{cond_exp}) if and only if 
\begin{equation}\label{birth_rate}
\lambda(t)=\alpha (1+t)^{-(\beta+1)}, \qquad t\geq 0.
\end{equation}
\end{proposition}
\par
Hence, when the expression of $\lambda(t)$ is specified as in (\ref{birth_rate}), 
we have $\Lambda(t)={\frac{\alpha}{\beta}\left[1-(1+t)^{-\beta}\right]}$, so that 
$\int_0^\infty \lambda(x)dx=\frac{\alpha}{\beta}$. Such model thus describes births under limited 
resources, with inputs that decrease in time and lead to a saturation level for the process.
Indeed, according to \cite{Ricciardi}, in this case both the conditional mean and variance 
are strictly increasing, with finite limits 
\begin{equation}\label{mean_simple}
E_y(t)=ye^{\frac{\alpha}{\beta}\left[1-(1+t)^{-\beta}\right]}
\quad \stackrel{t\rightarrow +\infty}{\rightarrow}\quad 
ye^{\frac{\alpha}{\beta}}\equiv C,
\end{equation}
\begin{equation}\label{variance_simple}
V_y(t)=y e^{\frac{\alpha}{\beta}\left[1-(1+t)^{-\beta}\right]}\left(e^{\frac{\alpha}{\beta}
\left[1-(1+t)^{-\beta}\right]}-1\right)
\quad \stackrel{t\rightarrow +\infty}{\rightarrow} \quad 
ye^{\frac{\alpha}{\beta}}\left(e^{\frac{\alpha}{\beta}}-1\right).
\end{equation}
%
%
Note that both $E_y(t)$ and $V_y(t)$ are decreasing in $\beta>0$. 
This is confirmed by  Fig.\ 9, where 
the mean and variance of $X(t)$ are plotted for various choices of $\beta$. 
For this special model, differently from the previous cases, we can determine a closed form 
for the {\em index of dispersion}, also known as {\em Fano factor}, defined as the variance 
over mean. Indeed, from Eqs.\  (\ref{mean_simple}) and (\ref{variance_simple}) we have 
\begin{equation}\label{eq:Dyt}
 D_y(t):=\frac{V_y(t)}{E_y(t)}=
 e^{\frac{\alpha}{\beta}\left[1-(1+t)^{-\beta}\right]}-1
 \quad \stackrel{t\rightarrow +\infty}{\rightarrow} \quad 
 e^{\frac{\alpha}{\beta}}-1,
\end{equation}
so that the index of dispersion  is monotonic increasing in $t$, with $D_y(0)=0$. 
From (\ref{eq:Dyt}) we have the following: 
\\
(i) \ If $\alpha<\beta \log 2$ then the birth process 
$X(t)$ is underdispersed, i.e.\ $D_y(t)<1$ for all $t\geq 0$. 
In this case the occurrence of events (births) is more regular than that of a Poisson process. 
\\
(ii) \ If $\alpha>\beta \log 2$ then $X(t)$ is  underdispersed for $t<t_*$, with 
$t_*:=\left(1-\frac{\beta}{\alpha}\ln 2\right)^{-1/\beta}-1$, and $X(t)$ is overdispersed for $t>t_*$, 
i.e.\ for large times there is more irregularity in the distribution of the number of occurrence 
of events (births) with respect to a Poisson process. 
\par
Moreover, from (\ref{mean_simple}) and (\ref{variance_simple}) we have that 
the {\em coefficient of variation} of $X(t)$ is given by
$$
\sigma_y(t)= \frac{\sqrt{V_y(t)}}{E_y(t)}
=y^{-1/2} \sqrt{1- e^{-\frac{\alpha}{\beta}\left[1-(1+t)^{-\beta}\right]}}, \qquad t\geq 0.
$$
Hence, it follows that $\sigma_y(t)$ is increasing in $t$ and in $\alpha$, whereas it is 
decreasing in $\beta$. The following limits thus hold: 
$$
\sigma_y(t) \quad \stackrel{t\rightarrow +\infty}{\rightarrow} \quad 
y^{-1/2} \sqrt{1-e^{-\frac{\alpha}{\beta}}},
$$
$$\lim_{\alpha \rightarrow 0} \sigma_y(t)=0,
\qquad \lim_{\alpha \rightarrow +\infty} \sigma_y(t)=y^{-1/2},$$
$$
 \lim_{\beta \rightarrow 0} \sigma_y(t)=y^{-1/2}\sqrt{1-(1+t)^{-\alpha}},
 \qquad \lim_{\beta \rightarrow+ \infty} \sigma_y(t)=0.
$$
Consequently, a better agreement between the deterministic growth model (\ref{new_sol}) and 
the birth process with rates (\ref{prob_j_simple}) is attained when $\alpha \rightarrow 0$ or 
$\beta \rightarrow + \infty$. 
Note that such conditions both imply $\lambda(t) \to 0$, so that a good correspondence 
between the two models is expected for a low birth rate. 
\begin{figure}[t]
\begin{center}
\epsfxsize=5.8cm
\epsfbox{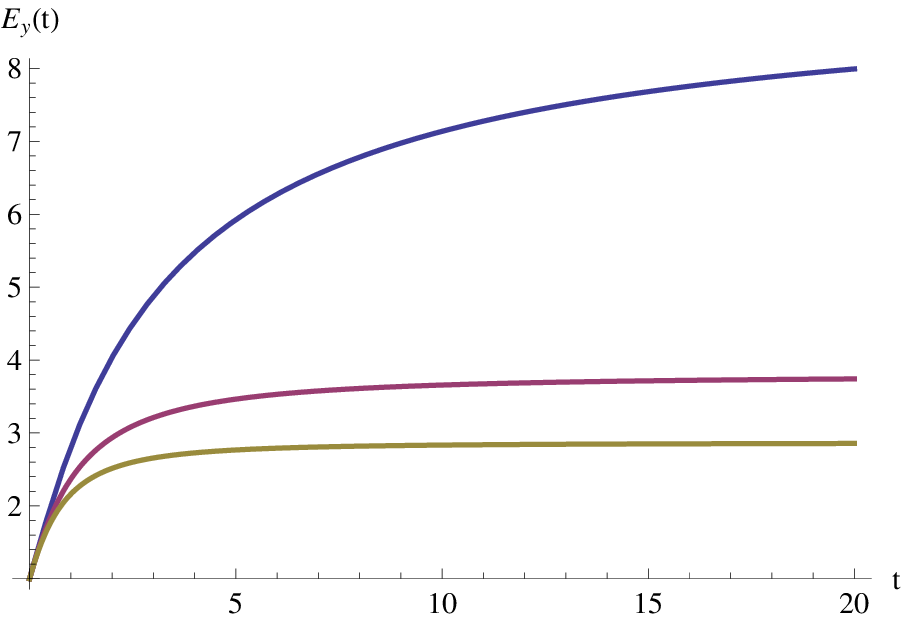}
\epsfxsize=5.8cm
\epsfbox{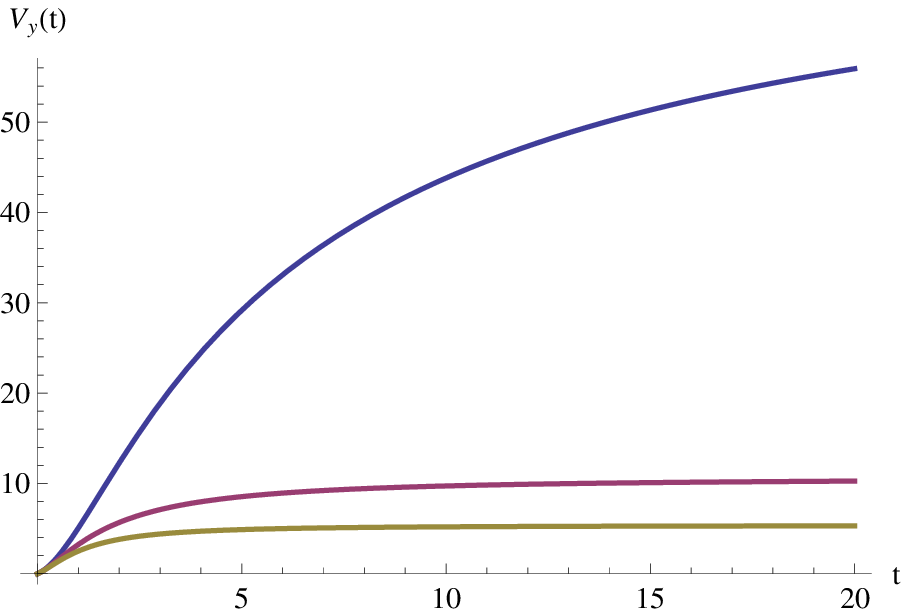}
\caption{The mean (\ref{mean_simple}) and the variance (\ref{variance_simple}) are 
plotted for $y = 1$, $\alpha = 2$, $\beta=0.9,\,1.5,\,1.9$ (from top to bottom). 
The limits (carrying capacity) for the mean are $C=9.23,\,3.79,\,2.87$, 
whereas for the variance the limits are $75.92,\,10.60,\,5.34$, respectively.}
\end{center}
\end{figure}
%
\par 
Let us now discuss a first-passage-time problem for $X(t)$, 
in analogy with the threshold crossing  problem analyzed in Section \ref{sec:Thresholdcrossing}.
The determination and analysis of  first-passage-time densities in biological modeling 
deserve large interest, since such functions provide essential information on the 
probability that some critical or beneficial levels are attained by the stochastic process under investigation. 
Given the initial state $y \in {\cal S}$ 
and a fixed threshold $k\in \mathbb{N}$, with $k>y$, we consider the first-passage time 
$$
T_{y,k}=\inf\{t\geq 0:\, X(t)=k\},  \qquad X(0)=y,
$$
and denote by $g_{y,k}(t)=dP(T_{y,k}\leq t)/dt$ its probability density function (pdf). 
Since the sample-paths of $X(t)$ are increasing over the state space $\cal S$, it is not 
hard to see that, for $k\in \mathbb{N}$, $k>y$, we have 
\begin{equation}\label{eq:fptpdf}
g_{y,k}(t)=(k-1)\lambda(t) P_{y,k-1}(t), \qquad t\geq 0, 
\end{equation}
where $\lambda(t)$ and $P_{y,k}(t)$ are given in (\ref{birth_rate}) and (\ref{trans_prob_simple}), 
respectively. Note that, since $\lambda(0)=\alpha$, the initial value of the first-passage-time pdf is 
$$
\lim_{t \rightarrow 0} g_{y,k}(t)=\left\{
\begin{array}{ll}
\alpha \,y, & \quad \hbox{if }k=y+1\\
0,& \quad \hbox{otherwise}.
\end{array} \right.
$$
In Fig.\ 10 the pdf of $T_{y,k}$ is plotted for various choices of $k$.   
We remark that such density is unimodal.
%
\begin{figure}[t]
\begin{center}
\epsfxsize=5.8cm
\epsfbox{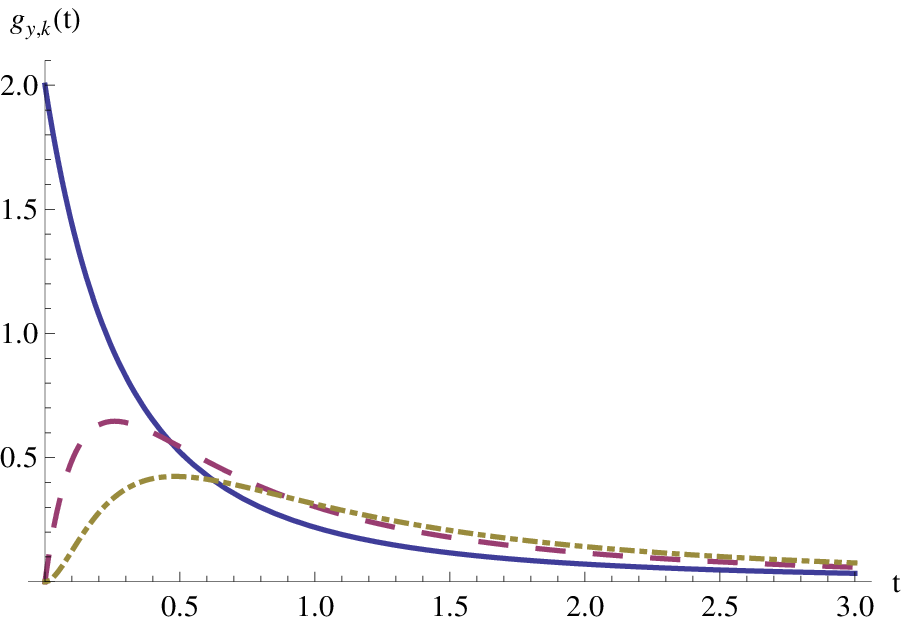}
\epsfxsize=5.8cm
\epsfbox{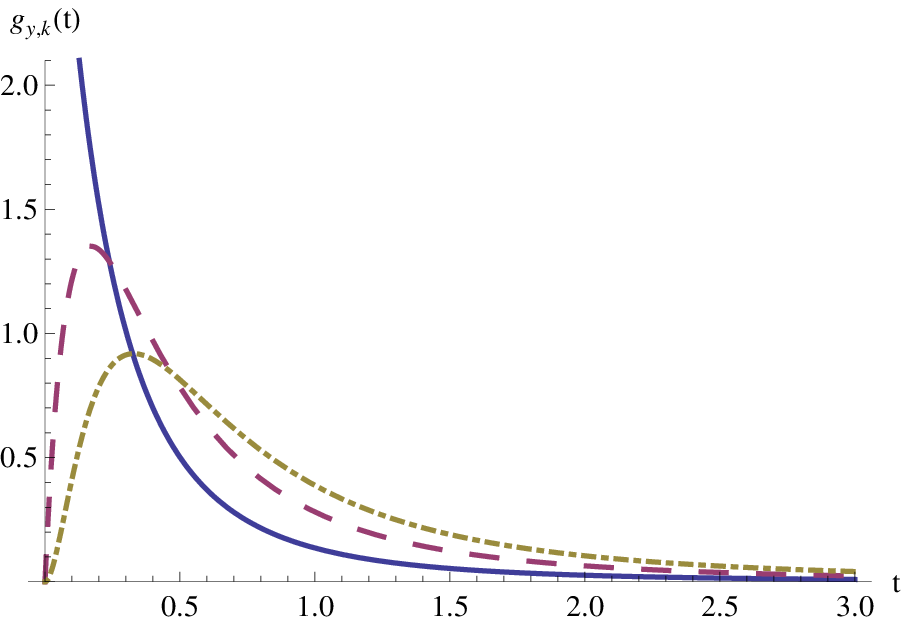}
\caption{
Density (\ref{eq:fptpdf}) for $\alpha=2,\, \beta=0.5$; on the left:  $k=2, \, 3, \, 4$ 
(solid, dashed, dot-dashed  line) and $y=1$; on the right:  $k=3,\,4,\,5$ 
(solid, dashed, dot-dashed  line) and $y=2$. 
}
\end{center}
\end{figure}
\begin{figure}[t]
\begin{center}
\epsfxsize=5.8cm
\epsfbox{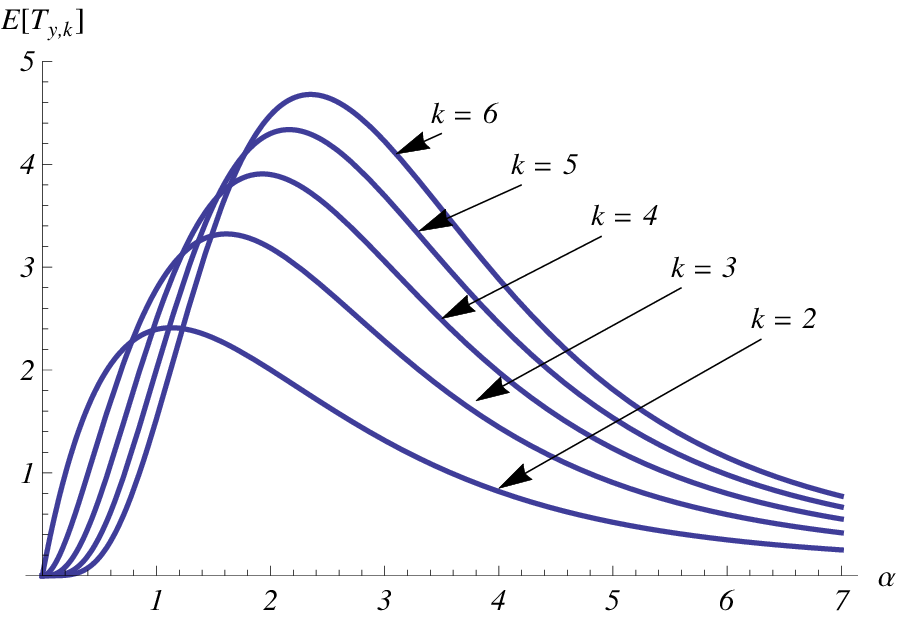}
\epsfxsize=5.8cm
\epsfbox{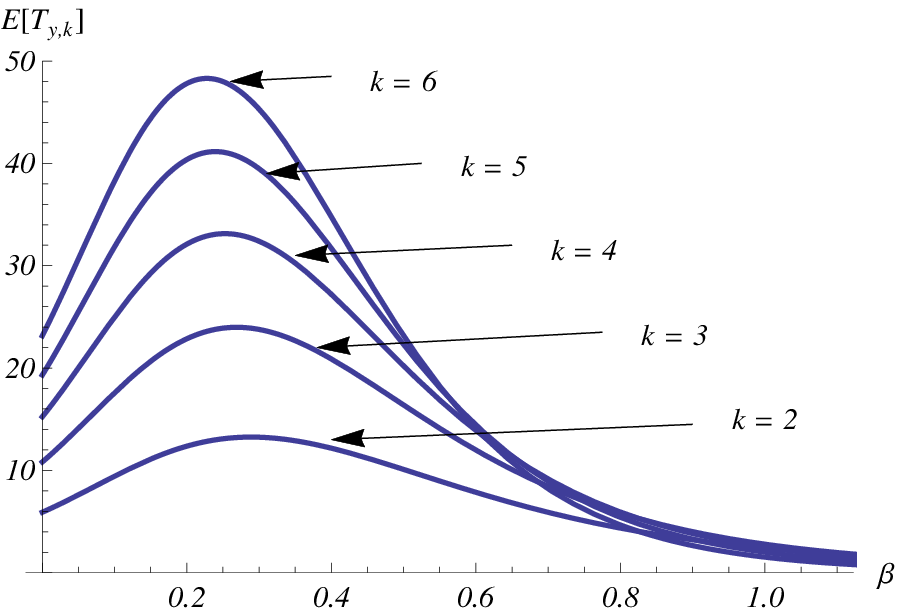}
\caption{For $y=1$, the mean of $T_{y,k}$ is plotted as function of $\alpha$ with $\beta=1$ 
on the left, and as function of $\beta$ with $\alpha=1$ on the right.}
\label{fig:11}
\end{center}
\end{figure}
%
In Fig.\ 11 the mean of $T_{y,k}$ is plotted as function of $\alpha$ and $\beta$, for various choices 
of $k$, evaluated by means of numerical integration. In both cases $E[T_{y,k}]$ is not 
monotonic, but it is increasing (decreasing) for small (large) values of $\alpha$ and $\beta$. 
\subsection*{\bf Concluding remarks} 
This paper has been devoted to the analysis of a deterministic growth model which is bounded by a 
carrying capacity, and behaves like the Gompertz law for small times, and the Korf law for large times. 
We investigated various properties of such a model, with attention to the main characteristics of 
interest in growth models, and performed various comparisons. 
Some examples of application of the proposed curve for the description of real 
growth phenomena have been shown, 
supported by an analysis of the performance based on suitable parameters. 
\par
In order to study a stochastic counterpart of the proposed model, we investigated a linear 
time-inhomogeneous birth-death process whose mean is identical to the proposed deterministic model. 
We obtained the transition probabilities, the moments and the population extinction of this process. 
However, even if the mean of such stochastic process is identical to the proposed growth curve, the 
considered birth-death model possesses an absorbing endpoint at zero and thus it may be not 
appropriate to describe a growth behavior. Hence, in order to overcome this drawback we finally  
studied a time-inhomogeneous birth process, whose mean again corresponds to the proposed 
growth curve, thus being more suitable to describe a pure growth phenomenon. 
The choices of the involved parameters that yield a better agreement between the deterministic 
growth model and the birth process have been identified though the analysis 
of the index of dispersion and the coefficient of variation. 

\subsection*{\bf Acknowledgements} 
This research was partially supported by the group GNCS of INdAM.  
The authors warmly thank the anonymous referees for various helpful comments. 


\end{document}